\documentclass[aps,pra,showpacs,reprint,footinbib,longbibliography,superscriptaddress]{revtex4-1}

\usepackage{etex}

\pdfoutput=1

\usepackage[greek,english]{babel}

\usepackage{amsmath,amsfonts,amssymb,amsthm,graphics,enumerate}
\usepackage{mathtools}

\usepackage[all]{xy}
\usepackage[dvipsnames]{xcolor}

\usepackage{soul}

\PassOptionsToPackage{hyphens}{url}

\usepackage[colorlinks=true,
            citecolor=blue,
            linkcolor=blue,
            urlcolor=black]
           {hyperref}

\usepackage{mdwlist}

\usepackage{enumitem}

\newtheorem{theorem}{Theorem}

\newtheorem{conjecture}{Conjecture}
\newtheorem{protocol}{Protocol}

\newtheorem{lemma}[theorem]{Lemma}

\newtheorem{definition}[theorem]{Definition}

\theoremstyle{definition}

\newcommand*{\mot}{~}

\newcommand*{\juan}{\color{black}}
\newcommand*{\je}{\textcolor{black}}
\newcommand*{\dom}{\textcolor{black}}

\newcommand*{\black}{\color{black}}

\newcommand{\bra}[1]{\langle #1|}
\newcommand{\ket}[1]{|#1\rangle}

\newcommand{\tr}{\mathrm{tr}}

\newcommand{\norm}[1]{\left\Vert #1 \right\Vert}

\newcommand{\uppi}{\text{\greektext{p}\latintext}}

\newcommand{\euler}{\textnormal{{e}}} 
\newcommand{\imun}{\textnormal{{i}}}  
\newcommand{\id}{I}

\newcommand\ppoly[1]{\textnormal{poly$\left(#1\right)$}}

\usepackage{times}

\allowdisplaybreaks

\newcommand{\fu}{Dahlem Center for Complex Quantum Systems, Freie Universit{\"a}t Berlin, 14195 Berlin, Germany}
\newcommand{\ubc}{University of British Columbia, Department of Physics and Astronomy,
Vancouver, BC, V6T 1Z1, Canada}

\begin{document}

\title{Architectures for quantum simulation showing a quantum speedup}

\author{Juan Bermejo-Vega}
\address{\fu}

\author{Dominik Hangleiter}
\address{\fu}

\author{Martin Schwarz} 
\address{\fu}

\author{Robert Raussendorf}
\address{\ubc} 

\author{Jens Eisert}
\address{\fu}

\begin{abstract}
One of the main aims in the field of quantum simulation is to achieve a quantum speedup, often referred to as ``quantum computational supremacy'', referring to the experimental realization of a quantum device that computationally outperforms classical computers. In this work, we show that one can devise versatile and feasible schemes of two-dimensional dynamical quantum simulators showing such a quantum speedup, building on intermediate problems involving non-adaptive measurement-based quantum computation. In each of the schemes, an initial product state is prepared, potentially involving an element of randomness as in disordered models, followed by a short-time evolution under a basic translationally invariant Hamiltonian with simple nearest-neighbor interactions and a mere sampling measurement in a fixed basis. The correctness of the final state preparation in each scheme is fully efficiently certifiable. We discuss experimental necessities and possible physical architectures, inspired by platforms of cold atoms in optical lattices and a number of others, as well as specific assumptions that enter the complexity-theoretic arguments. This work shows that benchmark settings exhibiting a quantum speedup may require little control in contrast to universal quantum computing. Thus, our proposal puts a convincing experimental demonstration of a quantum speedup  within reach in the near term.
\end{abstract}

\maketitle
\vspace{-15pt}
\vspace{-5pt}
\section{Introduction}

Quantum devices promise to solve computational problems efficiently for which no classical efficient algorithm exists.  The anticipated device of a universal quantum computer would solve problems for which no  efficient classical algorithm is known, such as integer factorization~\cite{Shor-1994} and simulating many-body Hamiltonian dynamics \cite{Lloyd-1996}. However, the experimental realization of such a machine requires fault-tolerant protection of universal dynamics against arbitrary errors \cite{Shor96_FTQC,KRZ96ThresholdAccuracy,Aharonov08FTQC_ConstantError}, which
incurs in  prohibitive  qubit overhead \cite{Fowler12SurfaceCodesPracticalQC,Lekitsche16BlueprintMicrowaveIonQV} beyond reach in available quantum devices. This does not mean, however, that the demonstration of a computational quantum advantage is unfeasible with current technology. 

Indeed, in recent years, it has become a major milestone  in quantum information processing to 
identify and build a simple (perhaps  non-universal) quantum  device that offers a large (exponential or superpolynomial) computational speedup compared to classical supercomputers, no matter what.  The demonstration of such  an advantage based on solid complexity-theoretic arguments  is often referred to as 
``quantum computational supremacy'' \cite{preskill2013quantum}.
This important near-term goal still constitutes a significant challenge, as technological advances seem to be required to achieve it, as well as  significant efforts in theoretical computer science, physics, and the numerical study of quantum many-body systems: after all, intermediate problems have to be identified with the potential to act as vehicles in the demonstration of a quantum advantage, in the presence of realistic errors.

There already is evidence that existing dynamical quantum simulators \cite{BlochSimulation,1408.5148} 
have the ability to outperform classical supercomputers. Specifically the 
experiments of Refs.\ \cite{Trotzky,MBL2D,Emergence,PhysRevX.7.041047} 
using ultracold atoms strongly suggest such a feature: They probe situations in which 
for short times \cite{Trotzky} or in one
spatial dimension \cite{Emergence,MBL2D}, the system can be classically simulated
in a perfectly efficient fashion using tensor network methods, even equipped with rigorous error bounds. However, 
for long times \cite{Trotzky} or in higher spatial dimensions \cite{Emergence,MBL2D}
such a classical simulation is no longer feasible with state-of-the-art simulation tools. 
Still, taking the role of devil's advocate, one may argue that this could be a consequence of a lack of imagination,
as there could---in principle---be a simple classical description capturing the observed phenomena.
Hence, a complexity-theoretic demonstration of a quantum advantage of quantum simulators
outperforming classical machines is highly desirable \cite{Roadmap}. Not all physically meaningful quantum simulations are to be be underpinned by such an argument, but it goes without saying that the field of quantum simulation 
would be seriously challenged if such a rigorous demonstration was out of reach.

Several settings for achieving a quantum speedup have been proposed\mot{}\cite{AaronsonArkhipov13LinearOptics4,BremnerPRL117.080501,Boixo161608.00263,Gao17SupremacyIsing,SparseNoisySupremacy} based on quantum processes that are classically hard to simulate probabilistically 
unless the Polynomial Hierarchy (\textsf{PH}) collapses. These processes 
remain hard to be simulated up to realistic
(additive) errors assuming further plausible complexity-theoretic conjectures. The proof techniques used build upon
earlier proposals giving rise to such a collapse \cite{terhal2004adaptive,BremnerJozsaShepherd08}.
However, at the same time they still come along with substantial experimental challenges.

This work constitutes a significant step towards identifying physically realistic
settings that show a quantum speedup  by laying out a
versatile and feasible family of architectures based on quenched local many-body dynamics. 
We remain close to what one commonly conceives as a dynamical quantum simulator 
\cite{BlochSimulation,CiracZollerSimulation,1408.5148,Trotzky},
\je{set up to probe exciting physics of interacting quantum systems.}
Indeed,
it is our aim is to remain as close as possible to experimentally-accessible or at
least realistic prescriptions, closely reminiscient of dynamical quantum simulators while at the same not compromising the rigorous
complexity-theoretic argument. 

Our specific contributions are as follows.
We focus on schemes in which random initial states are prepared on the 2D
square lattices of suitable periodicity,
followed by quenched, constant-time dynamics under a local
\emph{nearest-neighbor} (NN), \emph{translation-invariant} (TI)  Hamiltonian.
\je{These are prescriptions that are close to those that can be routinely
implemented with cold atoms in optical lattices}
\cite{BlochSimulation,Endres-apb-2013,ControlledCollisionsExperiment,Greiner}.
 Since evolution time is short, decoherence will be comparably small. 
In a last step, all qubits are measured in a fixed identical basis, producing
an outcome distribution that is hard to classically sample from within constant
$\ell_1$-norm error, requiring no postselection. Technically, our results implement sampling over new families of NNTI two-local
constant-depth \cite{terhal2004adaptive} IQP circuits
\cite{ShepherdBremner09_Temporally_Unstructured_QC,BremnerJozsaShepherd08,BremnerPRL117.080501}.
We build upon and develop a type of setting \cite{Hoban14PhysRevLett.112.140505} in which 
resource-states for measurement-based quantum computation are prepared
(MBQC)~\cite{RaussendorfPhysRevLett.86.5188}, but subsequently non-adaptively measured. 
We lay out the complexity-theoretic assumptions made, detail how they are
analog to those in Refs.~\cite{AaronsonArkhipov13LinearOptics4,BremnerPRL117.080501,Boixo161608.00263,Gao17SupremacyIsing,SparseNoisySupremacy}, and present results on anti-concentration.

By doing so, we arrive at surprisingly flexible and simple NNTI  quantum simulation schemes on square lattices, requiring different kinds of translational invariance in the preparation.  
Interestingly, and possibly counterintuitively, our schemes share the feature that the final state before the readout step can be efficiently
and rigorously certified in its correctness. This is further achieved via simple protocols that involve on-site measurements and a \dom{number of samples of the resource state that scales quadratically in the system size.} \je{The possibility of certification is unique
to our approach.} \dom{In fact, from the quadratically many samples of the prepared state, one can directly and rigorously 
infer about the very quantity that is used in the complexity-theoretic argument.} \je{We believe that the possibility
of such certification is crucial when it comes to unambiguously argue that a quantum device has the potential to show a true quantum speedup.}

Based on our analysis, we predict that short-time certifiable quantum-simulation experiments on as little as $50 \times 50$ qubit square lattices should be intractable for state-of-the-art classical computers \cite{Boixo161608.00263,Bravyi16ImprovedSimulationClifford+T}. It is important to stress that this assessment includes the rigorous certification part, and no hidden or unknown costs have to be added to this. Our proposed experiments are particularly suited to qubits arranged in two-dimensional lattices
\je{e.g., cold atoms in optical lattices \cite{BlochSimulation,Endres-apb-2013,ControlledCollisionsExperiment,Greiner}
with qubits encoded in hyperfine levels of atoms. When assessing feasible quantum devices,
it is crucial to emphasize that systems sizes of the kind discussed here are not larger but generally smaller
than what is feasible in 
present-day architectures \cite{BlochSimulation,Trotzky,Greiner,Emergence}.}


\section{Basic setup of the quantum simulation schemes}\label{setting}

We present a new family of simple physical architectures that cannot be efficiently simulated by classical computers with high evidence (cf.\ Theorem \ref{thm:Main} below). All share the basic feature that they are based on the constant-time evolution (quench) of an NNTI Hamiltonian  on a square lattice. 
Each architecture involves three  steps: 
\begin{enumerate}[label=\textnormal{Q\arabic*},leftmargin=17pt]
\item \label{Exp:Preparation} \textbf{Preparations.} 
Arrange $N:=\mu mn$ qubits side-by-side on an $n$-row $m$-column square lattice $\mathcal{L}$, with vertices $V$, edges $E$, initialized on a  product state 
\begin{equation}\label{eq:InitialState}
\ket{\psi_\beta}=\bigotimes_{i=1}^N \left( \ket{0}+\euler^{\imun\beta_i} \ket{1} \right),\: \beta\in \left\{0, \theta \right\}^N,
\end{equation}
for fixed 
$\theta\in\{\frac{\uppi}{4},\frac{\uppi}{8}\}$,
 which is chosen uniformly or randomly with probability $p_\beta$ (e.g., as a ground state of a disordered model). We consider standard  square primitive cells. We allow in one scheme each vertex to be equipped
with an additional qubit, named ``dangling bond qubit''. For this, $\mu{=}2$, otherwise $\mu{=}1$.
\item \label{Exp:Evolution} \textbf{Couplings.} Let the system evolve for  constant time $\tau=1$ under the effect of an NNTI Ising Hamiltonian 
\begin{align}\label{eq:HardIsingModels}
H:= \sum_{(i,j)\in E}J_{i,j} Z_iZ_j - \sum_{i\in V} h_i Z_i.
\end{align} 
This amounts to what is usually referred to as a \emph{quench}. Local fields $\{h_i\}_{i}$ and couplings $\{J_{i,j}\}_{i,j}$ are set to implement a unitary $U:=\euler^{\imun H}$, giving rise to a final ensemble $\Psi:=\{p_\beta, \ket{\Psi_\beta}\}_\beta$, $\ket{\Psi_\beta}:=U\ket{\psi_\beta}$. 
\item \label{Exp:Measurement} \textbf{Measurement.} Measure  primitive-cell qubits on the $X$ basis and (if present) dangling-bond qubits on the $Z$ basis. 
Since the latter can be traded for a measurement in the $X$ basis by a uniform basis rotation, one can equally well  measure all qubits in the same basis.
\end{enumerate}
\begin{figure}
    \includegraphics[width=\linewidth]{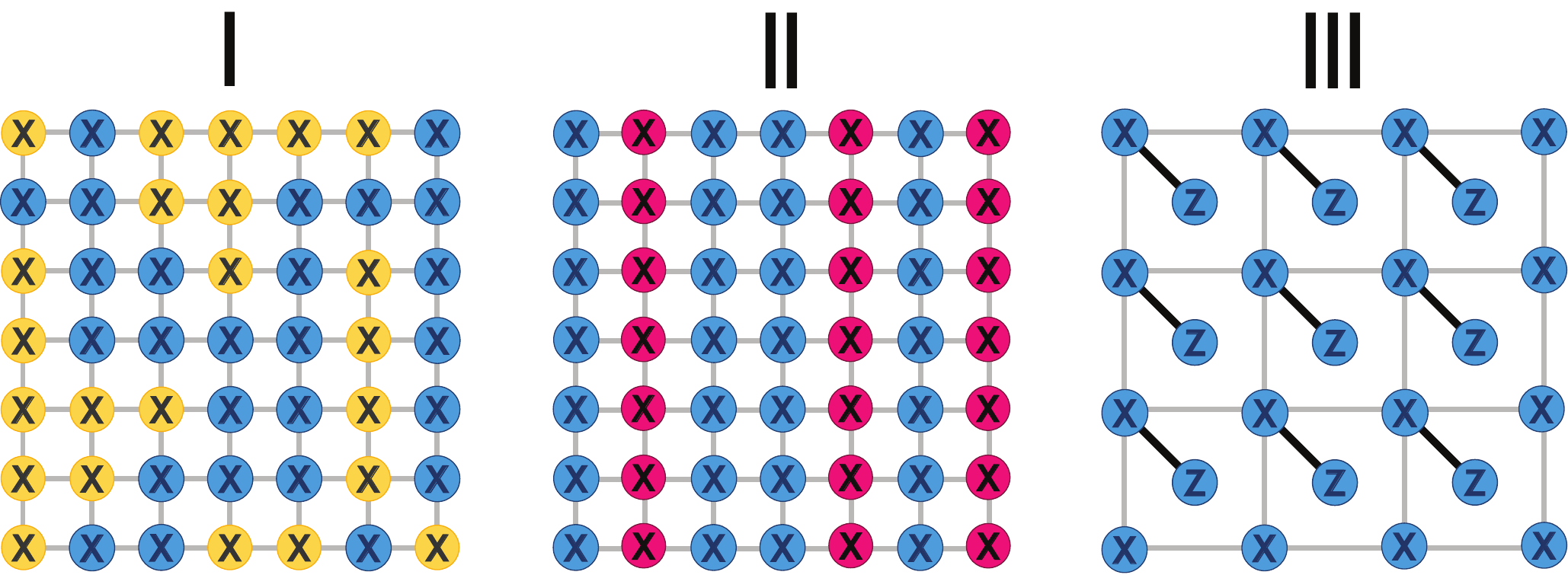}
\caption{Architectures \ref{enum:Arc1}-\ref{enum:Arc3}.
Colors illustrate the rotation angle of the initial state \eqref{eq:InitialState}: $\beta_i = 0$    (blue), $\beta_i = \pi/4$ (yellow), and $\beta_i = \pi/8$ (crimson). Solid lines between qubits represent Ising-type interactions \eqref{eq:HardIsingModels} with coupling constants $J_{i,j} = \pi/4$ (gray) and $J_{i,j} = \pi/8$ (black).   $X$  and $Z$ label the basis in which the respective qubits are to be measured.}\label{fig:Experiment}
\end{figure}
As we will discuss later, all individual steps have been realized with present technology. Note also that \ref{Exp:Evolution} amounts to a constant 
depth quantum circuit \cite{CommutingComment}.

\subsection{Physical desiderata and concrete schemes}

\begin{table*}
  \centering
   \begin{tabular}{| c|c | c | c |c| }
    \hline
    Scheme & Geometry  & Preparations & Couplings  & Measurements \\
    \hline
    \ref{enum:Arc1} & Square lattice  & ${\rm DO}$  & \{${\frac{\uppi}{4}}$\} $_{(1,1)}$ & $\{X\}_{(1,1)}$ \\
     \ref{enum:Arc2} & Square lattice  & TI$_{(1,\infty)}$ &  \{${\frac{\uppi}{4}}$\}$_{(1,1)}$ & $\{X\}_{(1,1)}$ \\
    \ref{enum:Arc3}  & Dangling-bond  square lattice  & TI$_{(1,1)}$ & \{${\frac{\uppi}{4},\frac{\uppi}{16}}$\}$_{(\sqrt{2},\sqrt{2})}$ & 
     $\{X, Z\}_{(\sqrt{2},\sqrt{2})}$ \\\hline

     \multicolumn{5}{|c|}{Previous work}\\
     \hline
     Ref.\ \cite{Gao17SupremacyIsing} & 7-fold brickwork graph   & TI$_{(1,1)}$ & \{${\frac{\uppi}{4}}$\}$_{(56,2)}$ & $\{X,X_{\pm\frac{\uppi}{4}},X_{\pm\frac{\uppi}{8} } \}_{(7,1)}$   \\
     \hline
   \end{tabular}
 \caption{Resource requirements of our hard-to-simulate quantum architectures. Architectures which employ simpler (more ordered) initial states require lattices of higher periodicity and finer controlled-rotations. The degree of symmetry of the preparation, evolution and measurement steps is quantified by the 2D periods indicated by vector subscripts $(a,b)$. We compare our results to the simplest previously known quantum simulation architecture on a planar graph showing a quantum speedup. The underlying complexity-theoretic assumptions needed for theses speedups are compared in section \ref{main result}. Above, $X_{\theta}:=\euler^{-\imun\frac{\theta}{2} Z}X\euler^{\imun\frac{\theta}{2} Z}$.  \label{tab:resources}}
 \end{table*}

Before  presenting  concrete schemes, we lay out desiderata that we associate with feasible ones. We require the implementation of each  step
 \ref{Exp:Preparation}-\ref{Exp:Measurement} to be as simple as possible. For preparations, couplings, 
 and measurements, we desire the periodicity as measured by the 2D periods $(k_x,k_y)$ in the $xy$ axes of the TI symmetry to be small.
Couplings should further be simple and not to scale with the system size. Last, we want the final measurement to be  translationally invariant.

We now present three concrete quantum architectures of the form \ref{Exp:Preparation}-\ref{Exp:Measurement}  that live up to the above desiderata. We label them\mot{}\ref{enum:Arc1}-\ref{enum:Arc3} and illustrate them in Fig.\mot{}\ref{fig:Experiment}:
\begin{enumerate}[label=\textnormal{\Roman*},leftmargin=17pt]
    \item\label{enum:Arc1} A disordered (DO) product state is prepared on a squared
    lattice, \black{} followed by a quench with an Ising Hamiltonian with
    couplings $J_{i,j}=h_i=\uppi/4$ \black{}---which implements controlled-$Z$
        ($CZ$) gates on edges---and a final measurement in the $X$ basis.
    \item \label{enum:Arc2} The initial state is TI  with period 1 in one lattice direction and uniformly random in the other (TI$_{(1,\infty)}$); couplings and measurements are picked as in \ref{enum:Arc1}. 
\item \label{enum:Arc3} Qubits are prepared on a dangling-bond square lattice. The initial state is TI with period 1 in all directions (TI$_{(1,1)}$). We pick $J_{i,j},h_i$ as in \ref{enum:Arc1}-\ref{enum:Arc2} on bright edges and  $J_{i,j}=h_i= \uppi/16$ on dark ones---the latter implement controlled-$\euler^{-\imun\uppi/8 Z}$\black{} ($CT$) gates on dangling bonds. Measurements are in the $Z$ basis for dangling qubits, elsewhere in the $X$ basis.
\end{enumerate}
(Cf.\ Appendix \ref{A} for full Hamiltonian  specifications.) 
The resources needed in each architecture are summarized in Table~\ref{tab:resources}. It is worth noting that, in all architectures \ref{enum:Arc1}-\ref{enum:Arc3} the state prepared after \ref{Exp:Evolution} is a resource for postselected measurement-based quantum computation  
postselecting w.r.t.\ the measurements in \ref{Exp:Measurement} (section \ref{D}), but as such does not amount to universal quantum computation. In fact,   architectures \ref{enum:Arc1}-\ref{enum:Arc3} require  \emph{neither} adaptive measurements (which are key in  MBQC~\cite{RaussendorfPhysRevLett.86.5188}) \emph{nor} physical postselection: our result states that if three plausible complexity-theoretic conjectures  hold  a single-shot readout cannot be classically simulated.\\

\section{Main result}\label{main result}

We now turn to stating that while the above three architectures \ref{enum:Arc1}-\ref{enum:Arc3} are physically feasible, the output distributions of measurements
cannot be efficiently classically simulated on classical computers, based on plausible assumptions and standard complexity-theoretic arguments.
\begin{theorem}[Hardness of classical simulation]\label{thm:Main}
If Conjectures\mot{}\ref{conj:ConjecturePolyHiear}-\ref{conj:AntiConcentration} below are true then a classical computer cannot 
sample from the outcome distribution of any architecture \ref{enum:Arc1}-\ref{enum:Arc3} up to error 1/22 in $\ell_1$ norm in time 
$O(\mathrm{poly}(n,m))$.
\end{theorem}
As in previous works \cite{terhal2004adaptive,AaronsonArkhipov13LinearOptics4,BremnerOld,jozsa2014classical,BremnerPRL117.080501,SparseNoisySupremacy,fefferman2016power}, Theorem~\ref{thm:Main} relies on plausible complexity-theoretic conjectures. The first, originally adopted in Ref.\  \cite{terhal2004adaptive}, is a widely believed statement about the structure of an infinite tower of complexity classes known as ``the Polynomial Hierarchy'' ($\textsf{PH}$), the levels of which recursively endow the classes \textsf{P}, \textsf{NP}, and \textsf{coNP} with oracles to previous levels.
\begin{conjecture}[Polynomial Hierarchy]\label{conj:ConjecturePolyHiear} 
The Polynomial Hierarchy is infinite.
\end{conjecture} 
The claim generalizes the familiar $\textsf{P}\neq\textsf{NP}$ conjecture in that $\textsf{P}=\textsf{NP}$ would imply a complete collapse of $\textsf{PH}$ to its $0$-th level. Furthermore, if two levels $k$, $k+1$ coincide, then all classes above level $k$ collapse to it. The available evidence for $\textsf{P}\neq\textsf{NP}$ makes Conjecture\mot\ref{conj:ConjecturePolyHiear} plausible, for it would be surprising to find a collapse of $\textsf{PH}$  to some level $k$ but not  a full one \cite{aaronson2016} (cf.\ Ref.\ \cite{fortnow2005beyond} for further discussion). Similarly to the Riemann hypothesis in number theory, many theorems in complexity theory have been proven relative to Conjecture \ref{conj:ConjecturePolyHiear}, probably most notably the Karp-Lipton theorem $\textsf{NP}\nsubseteq \textsf{P/poly}$\mot{}\cite{karp1980}. 

We highlight that, assuming  Conjecture ~\ref{conj:ConjecturePolyHiear} only, a classical computer would still not be able to sample from our experiments either exactly or within any  constant relative error (cf.\ section \ref{F}). However, such level of accuracy is physically unrealistic for it cannot be achieved by a quantum computer. A goal of this work is to understand how unlikely is for architectures \ref{enum:Arc1}-\ref{enum:Arc3} to be classically intractable  under realistic errors.

Our second conjecture, adopted from Ref.\ \cite{BremnerPRL117.080501}, is a qubit analog of the ``permanent-of-Gaussians'' conjecture \cite{AaronsonArkhipov13LinearOptics4}. It states that partition functions of (unstructured) random Ising models should be equally hard to approximate in average- and  worst-case. Now, let $(a,b):=(a_1,\ldots,a_{N_X},b_1,\ldots,b_{N_Z})$ be the outcomes of the  $X$ and $Z$ measurements in our architectures, with $b_i=0$  for \ref{enum:Arc1}-\ref{enum:Arc2}. In appendix \ref{B} we show that 
\begin{align}\label{eq:Probabilities=PartFunct}
\mathrm{prob}\left(a,b\,|\,\beta\right)&= \left|\bra{  a, b } U \ket{\psi_\beta}\right|^2=\frac{|\mathcal{Z}^{( \uppi a,\frac{\uppi}{4} b + \beta)}|^2}{2^{N_X+\frac{N_Z}{2}}},
\end{align}
where $\mathcal{Z}^{(\alpha,\vartheta)}:=\text{tr}(\euler^{\imun H^{(\alpha,\vartheta)}})$ is the partition function of a random Ising model on an $n\times m$ square lattice $\mathcal{L}_{\mathrm{sq}}$: \begin{align}\label{eq:RandomIsingModel}
 H^{(\alpha,\beta)}&:= \sum_{(i,j)\in E_{\mathrm{sq}}}\tfrac{\uppi}{4} Z_iZ_j - \sum_{i\in V_{\mathrm{sq}}} h_i^{(\alpha,\vartheta)} Z_i,
\\ 
h_i^{(\alpha,\vartheta)}&:= h_i- \left(\tfrac{\alpha_i+\vartheta_i}{2}\right), \quad \alpha_i\in \{0,\uppi\},  \vartheta_i\in\{0,\theta\},\notag
\end{align}where $\theta\in\{\tfrac{\uppi}{4},\tfrac{\uppi}{8}\}$ is chosen  as in step (\ref{Exp:Preparation}),\black{} and $\alpha$ (resp.\ $\vartheta$) is random and DO- (resp.\ either DO- or TI$_{(1,\infty)}$-) distributed. 
\begin{conjecture}[Average-case complexity]\label{conj:IsingAverageComplexity}   For random Ising models as in ~(\ref{eq:RandomIsingModel}), approximating  $|\mathcal{Z}^{(\alpha,\beta)}|^2$ up to relative error  ~$\frac{1}{4}+o(1)$ for any  0.3  fraction of the field configurations is as hard as in worst-case.
\end{conjecture}
We complement Conjecture\mot\ref{conj:IsingAverageComplexity} with the  following lemma.
\begin{lemma}[\textsf{$\#$P}-hardness]\label{lemma:SharpPHardness} Let $H^{(\alpha,\beta)}$ be the Ising model~(\ref{eq:RandomIsingModel}) on the $n\times m$ square lattice  with either  (i) \textnormal{DO}-distributed $\vartheta$ and $\theta\in\{0,\frac{\uppi}{4}\}$; or (ii) \textnormal{TI}$_{(1,\infty)}$-distributed $\vartheta$ and  $\theta\in\{0,\frac{\uppi}{8}\}$. Then, for $m\in O(n^2)$, approximating  $|\mathcal{Z}^{(\alpha,\beta)}|^2$ with relative error ~$\frac{1}{4}+o(1)$ is  \textnormal{\textsf{$\#$P}}-hard.
\end{lemma}
Thus, accepting Conjecture~\ref{conj:IsingAverageComplexity} implies that approximating $|\mathcal{Z}^{(\alpha,\beta)}|^2$ for these models is  as hard in average as any problem in \textsf{$\#$P}  \cite{valiant1979}. The proof  (section \ref{E}) applies MBQC methods \cite{RaussendorfPhysRevA.68.022312} to show that \ref{enum:Arc1}-\ref{enum:Arc3} are computationally equivalent to an encoded $n$-qubit 1D nearest-neighbor circuit comprising random gates of the form
\begin{equation}
\left[\prod_{i=1}^{n-1} CZ_{i,i+1}\right]\left[\prod_{i=1}^{n}Z_i^{c_i}\euler^{-\imun \theta d_iZ_i}H_i\right],c_i,d_i\in\{0,1\}\label{eq:TranslationInvariantGateset},
\end{equation}
where $c_i$ (resp.\ $d_i$) is DO (resp.\ DO-or-TI$_{(1,\infty)}$-)  distributed and $H$ is the Hadamard gate. Post-selecting such circuits, we can implement  two known universal  schemes of quantum computation \cite{Raussendorf05QC_TIOperations_1DChain,Mantri16MBQC_XYmeasurements}. We then exploit  that universal quantum-circuit amplitudes are \textsf{\#P}-hard to approximate.  As a remark, we discuss that the bound $m\in O(n^2)$ in Lemma~\ref{lemma:SharpPHardness} might not be optimal. In fact, we believe the result should still hold for $m\in O(n)$ (possibly for a different constant error) based on two pieces of evidence.
\begin{enumerate}[leftmargin=17pt]
\item[(i)]  On the one hand, our anti-concentration numerics (appendix \ref{G}), indicate that $O(n)$-depth universal random circuits of gates of form (\ref{eq:TranslationInvariantGateset})---whose output probabilities are in on-to-one correspondence via (\ref{eq:AntiConcentrationProb}) with the instances of $|\mathcal{Z}^{(\alpha,\beta)}|^2$---are Porter-Thomas distributed: the latter is a signature of quantum chaos, and of our quantum circuits being approximately Haar-random \cite{PorterThomas,HaakeChaos,EmersonPseudoRandom,Emerson05ConvergenceRandomCircuits,brown2008,Boixo161608.00263}. Hence, our numerics suggest that our $n\times n$-qubit lattices efficiently encode chaotic approximately-Haar-random $n$-qubit unitaries.
\item[(ii)] On the other hand, we analytically show that \textsf{\#P}-hardness arises for $m\in O(n)$ and  slightly- different choices of input states (resp.\ dangling-bonds) in architectures \ref{enum:Arc1}-\ref{enum:Arc2} (resp.\ \ref{enum:Arc3}) (cf. appendix \ref{H}).
\end{enumerate}
Last, we claim  that random circuits of  gates of form (\ref{eq:TranslationInvariantGateset}) anti-concentrate. 
\begin{conjecture}[Anti-concentration]\label{conj:AntiConcentration} Let    $\mathcal{C}$ be an $n$-qubit  $O(n)$-depth random circuit of gates of form (\ref{eq:TranslationInvariantGateset}), then \begin{equation}\label{eq:AntiConcentrationProb}
    \mathrm{prob}_{x}\left(\left|\bra{x} \mathcal{C}
    \ket{0}^{\otimes{n}}\right|^2\geq \frac{1}{2^{n}} \right) \geq \frac{1}{e}
\end{equation} 
    for a uniformly random choice of $x = (x_1,\ldots,x_n)$. 
\end{conjecture}
In section \ref{F}, we show that Eq.\mot(\ref{eq:AntiConcentrationProb}) is a sufficient condition for the output distribution of architectures  \ref{enum:Arc1}-\ref{enum:Arc3} to display anti-concentration. Analog numerically-supported conjectures  have been made in Refs.\ \cite{AaronsonArkhipov13LinearOptics4,Boixo161608.00263,Aaronson16_FoundationsQuantumSupremacy}. 
Here, we ran exact simulations of random circuits with up to 20 logical qubits to test Conjecture\mot{}\ref{conj:AntiConcentration} (appendix \ref{G}), and observed, first, that the anti-concentration ratio of  Eq.\~(\ref{eq:AntiConcentrationProb}) quickly converges to $1/e$ with the system-size; and, second, that measurement outcomes are Porter-Thomas \cite{PorterThomas} (i.e., exponentially) distributed, which is a signature of chaotic Haar-random unitary processes
\cite{HaakeChaos,EmersonPseudoRandom,Emerson05ConvergenceRandomCircuits,brown2008,Boixo161608.00263}. 

Previously, Refs.\ \cite{BremnerPRL117.080501,SparseNoisySupremacy,Bremner17SupremacyReview} argued that anti-concentration of measurement outcomes on an $n \times n$ lattice should require $\Omega(n)$ physical depth on 2D NN layouts in order not to induce a violation of the counting exponential
time hypothesis\mot{}\cite{Brand_et_al:LIPIcs:2017:6942,Curticapean2015}. This contrasts with the constant-depth nature of our proposal. To clarify this discrepancy, we note that our numerical evidence for anti-concentration is for logical $n$-qubit 1D circuits~(\ref{eq:TranslationInvariantGateset}) of depth $O(n)$ (Fig.\ \ref{fig:circuits}), which implies anti-concentration of the corresponding constant-depth evolution on a lattice of size $n \times O(n)$. Because this encoding introduces a linear overhead factor $O(n)$ there is no contradiction with  Refs.\ \cite{SparseNoisySupremacy,Bremner17SupremacyReview}. More critically, the observed signatures of anti-concentration rule out a potential efficient classical simulation of our schemes via sparse-sampling methods \cite{Schwarz13_Sparse,SparseNoisySupremacy}.



We end this section with a remark: Closest to our work is the approach of Ref.\ \cite{Gao17SupremacyIsing}, which is also an NNTI non-adaptive MBQC proposal albeit with larger resource requirements (see Table \ref{tab:resources} for a comparison). 
\dom{ 
What is more, it requires a stronger hardness assumption with regards to the required level of approximation in that it introduces a variation of Conjecture \ref{conj:IsingAverageComplexity} with a less-natural inverse-exponential additive error (cf. Appendix\mot\ref{app:Conjectures}). 
}
Furthermore, Refs.\ \cite{Boixo161608.00263,SparseNoisySupremacy,Aaronson16_FoundationsQuantumSupremacy} gave non-TI schemes based on time-dependent NN random circuits acting on square lattices: the latter approaches require  less qubits, but also circuits of polynomial depth. In ours and in that of Ref.\ \cite{Gao17SupremacyIsing}, circuit depth is traded with ancillas and kept constant. For ours and that of Ref.\ \cite{Gao17SupremacyIsing}, efficient certification protocols also exist and can be used to determine if the experiment has actually worked, as discussed below.

\section{Efficient certification of the final resource states}\label{sec:certification} 

It is key to all schemes proposed that the correctness of the final
resource-state preparation in the quantum simulation can be efficiently and rigorously certified. Since the prepared state is the ground state of a gapped
and frustration-free parent Hamiltonian $H_{\text{parent}} = \sum_i h_i$,
Ref.~\cite{Hangleiter} gives a scheme involving local measurements only that
certifies the closeness of the prepared state $\rho$ to anticipated state
$\ket{\Psi_\beta}\bra{\Psi_\beta}$ immediately before measurement in terms of
an upper bound on the trace-distance $\Vert \ket{\Psi_\beta}\bra{\Psi_\beta} -
\rho \Vert_1$ \cite{Hangleiter} (see also Ref.\ \cite{Cramer-NatComm-2010}). 
This directly yields an upper bound
on the $\ell_1$-norm distance between the respective measurement outcome distributions. 

The key idea of the protocol of Ref.\ \mot\cite{Hangleiter} is to estimate the energy $E_\rho = \tr [ \rho
H_{\text{parent}}]$. 
This yields a fidelity witness $F(\rho, \ket{\Psi_\beta}{\bra{\Psi_\beta}}) \geq
1 - E_\rho/\Delta$, where $\Delta$ is the gap of $H_{\text{parent}}$. 
This can be done, for example, by measuring the local Hamiltonian terms
$h_i$, which are five- (six-) body observables for architectures \ref{enum:Arc1}-\ref{enum:Arc2} (\ref{enum:Arc3}).  
\juan Ref.\mot\cite{Hangleiter} showed that this approach requires  $O(N^2\log N)$ samples of the state preparation $\rho$ to estimate a single term $\langle h_i \rangle$ with polynomial accuracy. Hence,   the full certification protocol requires $O(N^3\log(N))$ independent preparations of $\rho$ and five- (six-) body measurements. Though this scaling is efficient, its supra-cubic time-scaling and the complexity of the local measurements could render it impractical for near-term experiments with thousands of atoms.\black

\juan{} We now introduce three optimizations to the protocol of Ref.\ \cite{Hangleiter}, analyzed in Appendix \ref{app:certification} (Lemmas\mot{}\ref{lemma:energy}-\ref{certification}). First, for any non-degenerate gapped local Hamiltonian with known ground-state energy $E_0\in O(1)$ and known gap $\Delta\in \Omega(1)$, we show that the sample complexity of the protocol can be reduced to $O(N^2)$ by exploiting parallel measurement sequences of commuting Hamiltonian terms. Second, we show how to implement this more resource-economical protocol using on-site measurements only by expanding the latter terms in a local product basis. Third, we introduce a few setting-dependent optimizations for the architectures \ref{enum:Arc1}-\ref{enum:Arc3} (Lemma\mot{}\ref{lemma:ArcHamiltonianDecomposition}), tailored to their underlying square lattice geometry, the explicit tensor product  structure of the parent Hamiltonian of their pre-measurement states (Appendix\mot\ref{app:parentHamiltonians}), and their translation-invariant symmetry. 
We emphasize that,  as in the sampling measurement step \ref{Exp:Measurement} of our architectures, the optimized certification protocol relies on on-site measurements only.\black

In the three afore-mentioned cases (Appendix\mot\ref{app:certification}), the certification measurement pattern inherits the initial symmetry of the preparation step (Table  \ref{tab:resources}), i.e., DO for architecture \ref{enum:Arc1}, TI$_{(1,\infty)}$ for architecture\mot{}\ref{enum:Arc2}, TI$_{(\sqrt{2},\sqrt{2})}$ for architecture\mot{}\ref{enum:Arc3}. For architectures \ref{enum:Arc1}-\ref{enum:Arc2}, our setting resembles a certification protocol for preparing a family of hypergraph states given in Ref.~\cite{miller_quantum_2017}, though states and measurements therein are asymmetric. 
\begin{figure}
\centering
\includegraphics[width=0.45\linewidth]{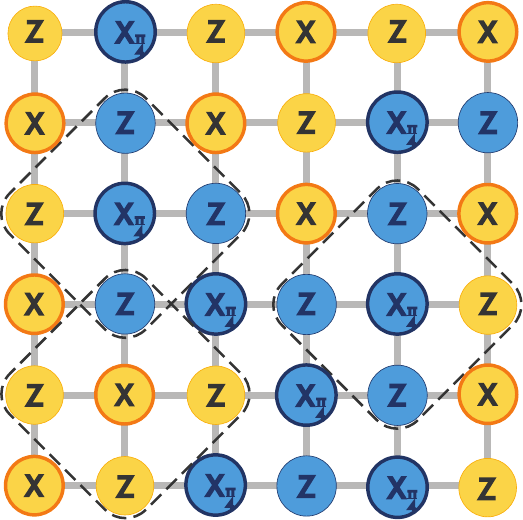}
\caption{\label{fig:certification} Certification protocol. We illustrate how our scheme works for architecture \ref{enum:Arc1}. Qubits with thick (thin) borders denote odd (even) sites in $V_\mathrm{odd}$ ($V_\mathrm{even}$). The figure illustrates a pattern of on-site measurements for one execution of the certification protocol that measures the energy of $H_\mathrm{odd}=\sum_{i\in V_\mathrm{odd}} h_i$ for the configuration of initial states in Fig.\ref{fig:Experiment}. On-site measurements are of type $Z, X$ and $X_{\uppi/4}$. Three hamiltonian terms whose joint measurement can be simulated from these single-qubit measurements are singled out by  the depicted dashed diamonds.}
\end{figure}

The above certification measurement is a slightly more difficult prescription than the experiments as such, yet as simple as one could hope, \juan    since we need to measure additional albeit single-qubit  bases\black. 
However, it is key to see that the correctness of the final-state preparation of an experiment can be certified even the absence of a known classical algorithm for sampling its
output distribution. This is also in contrast to other similar schemes, where no efficient rigorous scheme for certification of the final state before
measurement is known\mot{}\cite{AaronsonArkhipov13LinearOptics4,BremnerPRL117.080501,Boixo161608.00263,SparseNoisySupremacy,Aaronson16_FoundationsQuantumSupremacy}.

\je{Note that the certification protocol} \juan of Ref.\mot\cite{Hangleiter} is stated \black \je{in terms of noise-free measurements.
%
Nevertheless, for certain noise models it can be shown that rigorous certification is still possible \cite{Gao17SupremacyIsing}. 
Moreover, it is not an unreasonable assumption that local measurements can be benchmarked to very high a precision. 
Most importantly, the protocol readily accepts imperfect measurements: The imperfect measurements can 
be seen as perfect measurements, preceded by a quantum channel reflecting the noise.
This quantum channel noise can equally well be seen as acting on the quantum state, reducing the
trace-norm closeness to the anticipated state.} \dom{Hence, certifying $\epsilon$-closeness of the state-preparation under the assumption of ideal measurements is equivalent to certifying $\epsilon' < \epsilon$-closeness of the state preparation using measurements preceded by a noise channel with $\diamond$-norm at most $\epsilon - \epsilon'$.
To set up a detection scheme living up to the
required error bounds (that \je{scales} inversely with the system size for each on-site measurement) is demanding, but not unrealistic.}

Last, we highlight that the property of the final-state preparation being certifiable is rooted in the
fact that the states prepared are ground states of gapped frustration-free
local Hamiltonian models. 
At the same time they are injective \emph{projected
entangled pair states (PEPS)} of  constant bond dimension
\cite{Verstraete04VBS_for_QC,Models1}. The protocols discussed here can hence
be seen as PEPS sampling protocols that generate samples from local measurement
on PEPS.

\paragraph*{Certification protocol.}

Let us now outline the precise certification protocol (analyzed in Appendix\mot\ref{app:certification}) including the required
quantum measurements and the post-processing of the measurement outcomes. 
We do so in three steps: first, we find the parent Hamiltonians for state
preparations $\ket{\Psi_\beta}$ in architectures \ref{enum:Arc1}-\ref{enum:Arc3}. 
\juan Second, we show how 
on-site measurements are
sufficient to obtain a rigorous certificate. 
Finally, we  comment on the reduced sampling complexity $O(N^2)$ of the protocol. We refer the reader to Lemmas\mot{}\ref{lemma:energy}-\ref{certification}, Appendix \ref{app:certification}) for proofs of the results described below.\black

Observing that the resource states $\ket{\Psi_\beta}$ are stabilizer states,
all we need to do is find the appropriate stabilizers. The sum of the
stabilizers is then a parent Hamiltonian of $\ket{\Psi_\beta}$.  
For architectures \ref{enum:Arc1} and \ref{enum:Arc2} this yields (cf.\ App.~\ref{app:parentHamiltonians})
\begin{align}
    H_{\text{I,II}} & = -\sum_{i \in V}  \left( X_{\beta_i,i} \prod_{j:
(i,j) \in E} Z_j \right) \, , 
\label{stabilizer-1and2}
\end{align}
where $ X_{\beta_i,i} =\euler^{-\imun\frac{\beta_i}{2} Z}X_i\euler^{\imun\frac{\beta_i}{2} Z}$ is a rotated Pauli-$X$ operator
acting on site $i$ and $\beta_i$ is distributed as described
in \ref{Exp:Preparation}. 
Hence, the Hamiltonian consists of $N$ terms that are 5-local except at
the boundary of the lattice, where their locality is reduced to 4- or 3-local. 
In the specific case of architecture \ref{enum:Arc3}, 
a dangling-bond qubit is attached to each qubit via a $CT$ interaction. This yields a two-body term that replaces the
$X_{\beta_i,i}$ term in Eq.~\eqref{stabilizer-1and2} (see \mot{}App.~\ref{app:parentHamiltonians}).

The stabilizers $h_i = X_{\beta_i,i} \prod_{j:(i,j) \in E} Z_j$ in architectures\mot\ref{enum:Arc1}-\ref{enum:Arc2} can be measured using on-site measurements  in a demolition fashion, by first measuring  their tensor components and then multiplying their outcomes using classical post-processing. (This  procedure is reminiscent of the syndrome measurement of subsystem codes \cite{Poulin} and its correctness can easily be seen by decomposing each stabilizer into a eigenbasis.) 
\juan These act on distinct sites and can therefore be measured simultaneously. 
For the specific case of architecture \ref{enum:Arc3}, the local Hamiltonian terms  have on-site $Z$ factors as well as a two-body
component $CT X CT^\dagger$. \juan However, the purpose of measuring these stabilizers in the certification protocol of Ref.\ \cite{Hangleiter} (see Appendix \ref{app:certification}), is to estimate the average energy of the relevant parent Hamiltonian. To this end, we can,  w.l.o.g., expand $CT X CT^\dagger$ as a sum of product operators (Eq.\mot\ref{eq:FactorsofTwoBodyTerm}), measure the on-site factors appearing in this sum, and infer the target expected value using efficient classical post-processing. Hence, for our three architectures,  on-site measurements suffice.


Finally, the sampling complexity can be reduced from supra-cubic to quadratic $O(N^2)$ 
by simultaneously measuring commuting stabilizers (directly, or reducing them to on-site measurements as outlined above) on the same state preparation, following a specific pattern (Fig.~\ref{fig:certification}).\black{}  Precisely, we can define a lattice 2-coloring $V=V_{\mathrm{odd}}\cup V_{\mathrm{even}}$ and
simultaneously measure $Z$ on all sites $i\in V_{\mathrm{odd}}$ and $X_{\beta_j,j}$ on every site $j\in V_{\mathrm{even}}$ (or vice-versa). 
Since our Hamiltonian is commuting, each measurement round allows us to sample from the output distribution of $H_{\mathrm{odd}}:=\sum_{i\in V_\mathrm{odd}}h_i$ ($H_{\mathrm{even}}:=\sum_{i\in V_\mathrm{even}} h_i$), as shown in Fig.~\ref{fig:certification}. Hence, roughly $\sim N/2$ terms of the form $h_i$ can now be measured in parallel. \juan A simple application of Hoeffding's bound then shows that we can estimate the expected energy of $H_{\mathrm{odd}}$ ($H_{\mathrm{even}}$) using $O(N^2,)$ samples. \black{} Our proof concludes by noting that $\tr(H\rho) =\tr(H_{\mathrm{odd}}\rho)+\tr(H_{\mathrm{even}}\rho)$.

\section{Conceivable physical architectures}

We now turn to discussing that the above assumptions are plausible in several physical architectures close to what is available
with present technology.
What we are considering are large-scale quantum lattice architectures on square lattices ${\cal L}$ with a quantum 
degree of freedom per lattice site. 
On the level of physical implementation, the most advanced family of
such architectures and the most plausible \je{for the anticipated system sizes} is that provided by {\it cold atoms in optical lattices} \cite{BlochSimulation}. 
In an optical lattice architecture, internal degrees of freedom are available with \emph{hyperfine levels}. Also, the
encoding in spatial degrees of freedom within double wells is in principle conceivable.
Large-scale translationally invariant controlled-Z interactions---precisely of the type as they are
required for the preparation of cluster and graph states \cite{RaussendorfPhysRevA.68.022312,Graphs}---are feasible via controlled 
collisions \cite{ControlledCollisions,ControlledCollisionsExperiment}. \je{Actually,
the discussion of controlled collisions \cite{ControlledCollisions} and hence the theoretical
underpinning of such quenched dynamics triggered work on cluster states and measurement-based
quantum computing and predates this development.}
Other interactions to nearest neighbors are also conceivable. Interactions such as spin-changing collisions 
for $^{87}$Rb atoms have been experimentally observed \cite{PhysRevLett.95.190405}. Controlled-$T$ gates
require a more sophisticated interaction Hamiltonian. The dangling bonds seem realizable making use of
optical superlattices \cite{Nascimbene-PRL-2012,Trotzky}.
Single sites---specifically of the sampling type considered here---can be addressed in optical lattice architectures via several methods. E.g.,
quantum-gas microscopes allow for single-site resolved imaging \cite{Bakr,Kuhr}, even though the type of single-site 
addressing required here remains a significant challenge. Optical superlattices \cite{Nascimbene-PRL-2012,Trotzky} 
allow for the addressing of entire rows of sites in the same fashion. Entire rows that contain a fixed particle number
neighboring ones left empty can already be routinely prepared \cite{GrossPrivate}. Also, disordered initial
states can be prepared \cite{BlochMBL,MBL2D,PhysRevX.7.041047}.

But also other architectures are well conceivable. This includes in particular large 
arrays of {\it semiconductor quantum dots} allowing for single-site addressing---a setting that has 
already been employed to simulate the Mott-Hubbard model in the atomic limit \cite{Pellegrini}---or 
polaritons or exciton-polariton systems in {\it arrays of micro-cavities} \cite{PhysRevLett.112.116402}. In 
this type of architecture, the addressing of entire rows is also particularly feasible.
{\it Superconducting architectures} also promise to allow for large-scale array structures of the type anticipated here\mot{}\cite{Superconducting,SuperconductingSimulation,SuperconductingSimulation2}. 
{\it Trapped ions} can also serve as feasible architectures \cite{Trapped}. None of the physical architectures realize all elements required to the necessary precision, but at the same time, the prescriptions presented here are comparably close to what can be done.

\section{Proof of hardness result}

In this section we prove Lemma~\ref{lemma:SharpPHardness} and Theorem~\ref{thm:Main}, and develop the main techniques of the paper. The section is organized as follows:
\begin{itemize}[leftmargin=17pt]
\item In section \ref{C}, we use MBQC techniques to develop mappings that allow us to recast architectures \ref{enum:Arc1}-\ref{enum:Arc3} as (computationally equivalent) MBQCs on 2D cluster states (as introduced in  \cite{RaussendorfPhysRevLett.86.5188}).
\item In section \ref{D}, we show that enhancing architectures \ref{enum:Arc1}-\ref{enum:Arc3} with the ability (or an oracle) to post-select the outcomes of random variables turns them as powerful as a post-selected universal quantum computer (as defined  in\mot\cite{Aaronson-ProcRS-2005}). 
\item In section \ref{E}, we prove Lemma \ref{lemma:SharpPHardness} using earlier findings and a new parallelization technique to implement the two-local ``dense'' IQP circuits of Ref.\ \cite{BremnerPRL117.080501} in linear depth on a 1D nearest-architecture. 
\item In section \ref{F}, we give the proof of our main result Theorem\mot\ref{thm:Main}. The proof makes use of Lemma~\ref{lemma:SharpPHardness} and Stockmeyer's Theorem~\cite{Stockmeyer85ApproxiationSharpP}. The latter is applied in an analogous way as in the Boson-sampling and IQP-circuit settings\mot\cite{AaronsonArkhipov13LinearOptics4,BremnerPRL117.080501} to show that if an efficient classical algorithm can approximately sample from the output distribution of architectures \ref{enum:Arc1}-\ref{enum:Arc3}, then an \textsf{FBPP}$^{\textsf{NP}}$ algorithm can approximate a large fraction of the amplitudes in (\ref{eq:Probabilities=PartFunct}), if the latter are also sufficiently anti-concentrated. By Conjectures \ref{conj:IsingAverageComplexity}-\ref{conj:AntiConcentration}, the latter algorithm can solve any problem in \textsf{P}$^\textsf{\#P}$, which contains \textsf{PH} via Toda's theorem \cite{toda}. This implies a collapse of the Polynomial Hierarchy to its 3rd level. 
\end{itemize}

\subsection{Mapping architectures \ref{enum:Arc1}-\ref{enum:Arc3} to cluster state MBQCs}\label{C}

We show that any architecture \ref{enum:Arc1}-\ref{enum:Arc3} can be mapped via a bijection  to a computationally-equivalent sequence of X-Y-plane single-qubit measurements on the  2D cluster state\mot{}\cite{RaussendorfPhysRevLett.86.5188}. Below,  $T:=\mathrm{diag}\left(1,\euler^{\imun\uppi/4}\right)$ and $\sqrt{T}:=\mathrm{diag}\left(1,\euler^{\imun\uppi/8}\right)$.
First, note that (via teleportation)  the effect of measuring a dangling-qubit  (if present) is equivalent  to generating a uniformly-random classical bit $b\in\{0,1\}$ and, subsequently, implementing the gate $T^b$  onto its neighbor;  we can thus replace all dangling-bond qubits  by introducing a uniformly-random  measurement of  $X$ or $X_{-\frac{\uppi}{4}}=T^\dagger XT\propto X-Y$ on every primitive qubit. Furthermore, we can re-write the input  $\ket{\psi_\beta}$ in \ref{Exp:Preparation} as 
\begin{equation}\label{app:eq:InputState}
\ket{\psi_\beta}=\bigotimes_{i=1}^N \sqrt{T}^{kb_i}\ket{+}^{\otimes N},
\end{equation}
where $\ket{+}\propto\ket{0}+\ket{1}$, $k\in\{1,2\}$  and  $b=(b_1,\ldots,b_N)$ is a random bit-string defined via $b_i:=\beta_i/\theta$, with $\beta,\theta$ as in~\ref{Exp:Preparation}. Since $\sqrt{T}$ gates in (\ref{app:eq:InputState}) commute with the Hamiltonian~(\ref{eq:HardIsingModels}) and their effect is unobserved by $Z$ measurements, they can be propagated out of the experiment by measuring $X_{-\frac{\uppi}{8}}=\sqrt{T}^{b_i\dagger}X_i\sqrt{T}^{b_i}$ instead of $X$ on every  primitive qubit $i\in V$.
Combining these facts, we obtain the following mappings:
\begin{enumerate}[label=\textnormal{(C\arabic*)}]
\item Architectures \ref{enum:Arc1} and \ref{enum:Arc3} are computationally equivalent to a quantum circuit that prepares a 2D cluster state on their underlying primitive square lattice and measures $\{X,X_{-\frac{\uppi}{4}}\}$ randomly on each vertex.\label{app:MapMBQCA1A3}
\item Architecture \ref{enum:Arc2} is computationally equivalent to an analogous circuit of random  $\{X,X_{-\uppi/8}\}$ single-qubit measurements, which chooses the latter measurements to be identical along the columns of the 2D cluster state.\label{app:MapMBQCA2}
\end{enumerate}

\subsection{Universality of architectures \ref{enum:Arc1}-\ref{enum:Arc3} for postselected measurement based quantum computation}\label{D}

For each of our architectures \ref{enum:Arc1}-\ref{enum:Arc3}, we prove that the ensemble $\{p_\beta,\ket{\Psi_\beta}\}_\beta$ is a universal resource for postselected MBQC w.r.t.\ the measurements in step \ref{Exp:Measurement}. Precisely, this means that if the ability to post-select the outcomes of the experiment's random variables (the qubit outcomes in step \ref{Exp:Measurement} and the random vector $\beta$) is provided as an oracle \cite{terhal2004adaptive,Aaronson-ProcRS-2005}, then it is possible to implement any poly-size quantum circuit~\cite{Nielsen} with arbitrarily high-fidelity in a subregion of the lattice using  (at most) polynomially-many qubits. Our proof is constructive and shows how to simulate universal circuits of Clifford+$T$  gates \cite{Boykin00_Clifford+T} via postselection.

Below, we call a quantum circuit  \emph{1D homogeneous} if it consists of 1D  nearest-neighbor gates that and  parallel operations are identical modulo  translation. The latter do not need to be translation-invariant: e.g., an arbitrary $S$-size 1D nearest-neighbor circuit can be serialized to be 1D homogeneous in depth $O(S)$. In Fig.~\ref{fig:LinearDepthIQP} below, we give an example of an $S$-size IQP circuit that can be implemented in depth $O(\sqrt{S})$ (by bringing single-qubit gates to the end). 1D homogeneous circuits, as defined here, can be regarded as examples of quantum cellular automata \cite{Raussendorf05_QCA_universalQC}.
\begin{lemma}[Postselected universality] \label{lemma:PostSelection}  Let $V$ be an $n$-qubit  $D$-depth 1D homogeneous  circuit of Clifford+$T$ gates. Then, for any architecture \ref{enum:Arc1}-\ref{enum:Arc3}, it is possible to prepare the right-most primitive qubits of an $O(n)\times O(Dn)$-qubit lattice on a state $\ket{\psi}:=\left(V\ket{0}^{\otimes n}\right)\ket{0}^{\otimes r}$, $r\in O(n)$,  using postselection.
\end{lemma}
We highlight that the complexity of the simulation  in Lemma~\ref{lemma:PostSelection} scales with the depth of the input circuit (not the size), allowing to parallelize concurrent nearest-neighbor gates. To prove  this result, we  assumes basic knowledge of MBQC on cluster states~\cite{RaussendorfPhysRevLett.86.5188,RaussendorfPhysRevA.68.022312}. Additionally, we make use of two technical lemmas.\black{}
\begin{lemma}[Efficient preparation via MBQC]\label{lemma:lemma:MBQCX-YOnColumns} Let $V$ be an $n$-qubit $D$-depth 1D homogeneous  Clifford+$T$ circuit. Then, the state vector $\ket{\psi}:=\left(V\ket{0}^{\otimes n}\right)\ket{0}^{\otimes {3n-2}}$ can be efficiently prepared exactly via an MBQC of single-qubit $\{X,X_{\pm\uppi/8}\}$ measurements on an $(4n-2)\times O(Dn)$-qubit  2D cluster state, and even if measurements are constrained to act ``quasi-periodically'' as follows: for every column,  each of its qubits is measured in either the $X$ basis, or in one of the $X_{\pm\uppi/8}$ bases (where the sign can be picked freely on distinct sites). 
\end{lemma}
\begin{lemma}[On-site efficient preparation via MBQC] \label{lemma:MBQCX-YOnSites}
Let $V$ be an $n$-qubit $D$-depth 1D homogeneous Clifford+$T$ circuit. Then, the state vector $\ket{\psi}=V\ket{0}^{\otimes n}$ can be efficiently prepared exactly via an MBQC of single-qubit $\{X,X_{\pm\uppi/4}\}$ measurements on an $n\times O(Dn)$-qubit  2D cluster state.
\end{lemma}
Lemma \ref{lemma:lemma:MBQCX-YOnColumns} is an MBQC implementation of a 1D quantum-computation scheme given in Ref.\ \cite{Raussendorf05QC_TIOperations_1DChain}. Lemma~\ref{lemma:MBQCX-YOnSites} follows from Lemma~3 in Ref.\  \cite{Mantri16MBQC_XYmeasurements} by using that commuting-gate measurement-patterns can be applied simultaneously in MBQC.
\begin{proof}[Proof of Lemma \ref{lemma:lemma:MBQCX-YOnColumns}]
We first show how to implement a universal set of gates that can be converted to the Clifford+$T$ gate-set. We begin by picking a translationally-invariant gate set with the desired property \cite{Raussendorf05QC_TIOperations_1DChain}
\begin{align}
&\left\lbrace  E := \left(\prod_{i=1}^{M-1} CZ_{i,i+1}\right) \left(\prod_{j=1}^M H_i\right),\quad Y_\mathrm{all}:=\prod_{j=1}^M Y_i,\right.\notag \\
& \left.\:\:\:\:  U_A(\alpha):=\prod_{j=1}^N \euler^{-\imun \frac{\alpha}{2} A_i},\quad \textnormal{ where } A\in\{X,Z\} \hfill 
\right\rbrace \label{app:eq:TranslationInvariantGateset}.
\end{align}
Above,  gates act on a one-dimensional chain of $M:=4n-2$; $H$ is the Hadamard gate; $CZ_{i,i+1}$ is the $CZ$ gate on qubits $i,i+1$;  $E$ is a global entangling gate; and $E^{M+1}$ implements a ``mirror'' permutation $i\rightarrow \overline{i}:=M+1-i$  of the qubits. The computation is encoded on  $n$ logical qubits with physical positions $[i]:=2i-1,1\leq i \leq n$. The remaining qubits are kept in the state $\ket{0}$. Ref.\  \cite{Raussendorf05QC_TIOperations_1DChain} shows how to implement generators for the Clifford+$T$ gate-set using the following sequences of\mot{}(\ref{app:eq:TranslationInvariantGateset}) operations:
\begin{enumerate}[label=\textnormal{(S\arabic*)}]
\item $\left(E^{M+1-i}Y_\textnormal{all}EY_\textnormal{all}E^{i-1}U_Z\left(\mp\frac{\uppi}{8}\right)\right)\allowbreak\left(E^{M+1-i}Y_\textnormal{all}EY_\textnormal{all}E^{i-1}U_Z\left(\pm\frac{\uppi}{8}\right)\right)$,
\label{app:GateSequence1}
\item $\left(E^{M-i}Y_\textnormal{all}EY_\textnormal{all}E^{i}U_X\left(\mp\frac{\uppi}{8}\right)\right)$ $\left(E^{M-i}Y_\textnormal{all}EY_\textnormal{all}E^{i}U_X\left(\pm\frac{\uppi}{8}\right)\right)$,   \label{app:GateSequence2}
\item $\left(E^{M-2-[i]}Y_\textnormal{all}EY_\textnormal{all}E^{{[i]+1}}U_X\left(\mp\frac{\uppi}{8}\right)E\right)$ $\left(E^{M-2-[i]}Y_\textnormal{all}EY_\textnormal{all}E^{[i]+1}U_X\left(\pm\frac{\uppi}{8}\right)E\right)$. \label{app:GateSequence3}
\end{enumerate}
Sequence \ref{app:GateSequence1} implements a  $\euler^{\mp\imun\frac{\uppi}{8} Z_{i}}$ gate; \ref{app:GateSequence2},   a  $\euler^{\mp\imun \frac{\uppi}{8}X_{i}}$ gate; and \ref{app:GateSequence3}, a logical $\euler^{\mp\imun \frac{\uppi}{8} X_{[i]}X_{[i+1]}}$ gate \cite{Raussendorf05QC_TIOperations_1DChain}.
We will now show that any of the above gate sequences can be implemented directly on an MBQC on an $(4n-2)\times O(n)$-qubit 2D cluster state with quasi-periodic $\{X,X_{\pm\uppi/8}\}$ measurements. W.l.o.g., we make $E$ (resp.\  non-entangling unitaries) act on even steps (resp.\ odd ones) by introducing identity gates when necessary. We now  reorder operations in the creation and measurement of the cluster state as indicated in Fig.\ \ref{fig:MBQC_TIScheme}: therein, balls denote qubits prepared in  $\ket{+}$; steps Fig.~\ref{fig:MBQC_TIScheme}.(1) and Fig.~\ref{fig:MBQC_TIScheme}.(2) implement $CZ$ gates for the preparation of the cluster state; and Fig.~\ref{fig:MBQC_TIScheme}.(3) implements a round of measurements.

In MBQC, $Y_\textnormal{all}$ gates can be treated as byproduct Pauli operators and do not need to be enacted~\cite{RaussendorfPhysRevA.68.022312}. Further, performing a periodic  measurement of $X$ in Fig.~\ref{fig:MBQC_TIScheme}.(3) implements the $E$ gate in (\ref{app:eq:TranslationInvariantGateset}). In turn, a quasi-periodic $X_{\pm\uppi/8}$ measurement, where observables' signs are chosen adaptively to counteract random byproduct operators, can be used to implement $E$ followed by a $U_Z(\pm\uppi/8)$ gate. Similarly, and last, we can implement $EU_X(-\uppi/8)$ by delaying the measurement to the next step and propagating a $U_Z(-\uppi/8)$ backwards: this works because $U_Z$ and $U_X$ never occur in subsequent odd steps in sequences \ref{app:GateSequence1}-\ref{app:GateSequence2}-\ref{app:GateSequence3}. We thus have an MBQC simulation of the translation-invariant computation in Ref.\ \cite{Raussendorf05QC_TIOperations_1DChain}.
\begin{figure}
\includegraphics[width=0.35\linewidth]{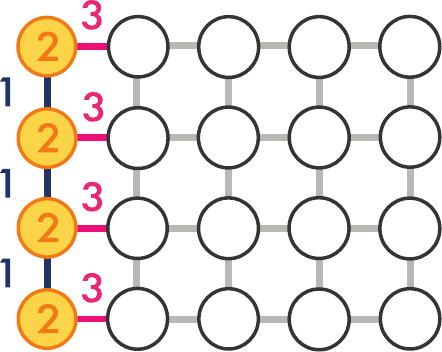}
\caption{Mapping the TI quantum computation scheme of Ref.\ \cite{Raussendorf05QC_TIOperations_1DChain} to a cluster-state MBQC.}
\label{fig:MBQC_TIScheme}
\end{figure}

Last, we show that an $n$-qubit $D$-depth homogeneous circuit of $\euler^{\mp\imun \uppi Z_{[i]}/8}$, $\euler^{\mp\imun \uppi X_{[i]}/8}$, $\euler^{\mp\imun \uppi X_{[i]}X_{[i+1]}/8}$ gates can be implemented on an $(4n-2)\times O(Dn)$-qubit cluster-state MBQC using the above protocol. Here, we invoke that measurement-patterns on  disjoint-regions of a cluster-state MBQC can be  simultaneously applied for  commuting logical gates (hence, also concurrent ones): the latter fact is easily verified in the MBQC's logical-circuit picture\mot\cite{RaussendorfPhysRevA.68.022312,Childs05Unified_MBQC_PRA}. 
\end{proof}
\begin{proof}[Proof of Lemma~\ref{lemma:MBQCX-YOnSites}]  Lemma 3 in Ref.\ \cite{Mantri16MBQC_XYmeasurements} shows that performing an $X_{\pm\uppi/4}$ measurement on a boundary qubit of an $n \times (n + 2)$ one can selectively implements any logical gate of form $\euler^{\mp\imun \uppi Z_{i}/8}$, $\euler^{\mp\imun \uppi X_{i}/8}$, $\euler^{\mp\imun \uppi Z_{j}X_{j+1}/8}$, $\euler^{\mp\imun \uppi X_{j}Z_{j+1}/8}$, $1\leq  i \leq n$, $1\leq  j \leq n-1$ on an $n$-qubit 1D chain. As in proof of Lemma~\ref{lemma:MBQCX-YOnSites}, measurement-patterns associated to commuting gates can be implemented simultaneously. The  result of Ref.\mot{}\cite{Mantri16MBQC_XYmeasurements} thus yields an exact  $n \times (n + 2)$ cluster-state MBQC implementation of any circuit of form
\begin{equation}
\prod_{i=1}^n \euler^{-\imun \frac{\uppi b_i}{8} X_{i}}   C \prod_{i=1}^n \euler^{-\imun \frac{\uppi a_i}{8} Z_{i}}, \quad b_i,a_i\in{0,1}
\end{equation} 
for any $n$-qubit 1D  commuting circuit $C$ of $\euler^{-\imun \frac{\uppi}{8} Z_{j}X_{j+1}}$,  $\euler^{-\imun \frac{\uppi}{8} X_{j}Z_{j+1}}$ gates. The proof follows as the one of Lemma\mot{}\ref{lemma:lemma:MBQCX-YOnColumns}.
\black{}
\end{proof}
We now proof the main claim of this section.
\begin{proof}[Proof of Lemma \ref{lemma:PostSelection}] Recall that our architectures can be re-casted as a non-adaptive cluster-state MBQCs via  mappings\mot{}\ref{app:MapMBQCA1A3}-\ref{app:MapMBQCA2}. Hence, it suffices to show how to prepare  $\ket{\psi}=(V\ket{0}^{\otimes n})\ket{0}^{\otimes r}$  exactly for some $r\in O(n)$ using two kinds of operations:
\begin{enumerate}[label=\textnormal{(D\arabic*)}]
\item Postselected random $\{X,X_{-\frac{\uppi}{4}}\}$  measurements on a $(n+r)\times O(Dn)$-qubit  cluster state\label{app:ProofPostSelectionLemma1}.
\item Postselected random $\{X,X_{-\frac{\uppi}{8}}\}$ measurements, chosen identically on columns, on a $(n+r)\times O(Dn)$-qubit  cluster state.\label{app:ProofPostSelectionLemma2}
\end{enumerate}
Statement  \ref{app:ProofPostSelectionLemma1} (resp.\ \ref{app:ProofPostSelectionLemma2}) covers the case for  \ref{enum:Arc1} and \ref{enum:Arc3} (resp.\ architecture~\ref{enum:Arc2}). 
To prove \ref{app:ProofPostSelectionLemma1}-\ref{app:ProofPostSelectionLemma2}, we  show that if an \textnormal{MBQC} scheme on a cluster state is universal w.r.t.\ a family of X-Y plane measurements $\{X_{\theta_i}\}_i$, $X_{\theta_i}=\euler^{-\imun\frac{\theta_i}{2} Z}X\euler^{\imun\frac{\theta_i}{2} Z}$, then, the reduced negative-angle subfamily $\{X_{-|\theta_i|}\}_i$ is universal for \textnormal{postMBQC}; in combination with Lemmas~\ref{lemma:MBQCX-YOnSites}-\ref{lemma:lemma:MBQCX-YOnColumns}, the claims follow. Recall that any non-final measurement in cluster-state MBQC\mot{}\cite{RaussendorfPhysRevLett.86.5188,RaussendorfPhysRevA.68.022312} produces a uniformly random outcome $s\in\{0,1\}$ (cf.\ section \ref{E}, Eq.~(\ref{app:eq:MarginalProbs}) for an explicit formula), whose effect in the logical circuit is to introduce a random byproduct Pauli operator $X^s$ on its associated qubit-line. If  not accounted for (e.g., by adapting the measurement basis) and an  $X_{\theta}$ is subsequently performed, the latter effectively implements a $X_{(-1)^s\theta}$ measurement: this can be seen by propagating $X^s$ forward in the circuit using conjugation relationships, and it is illustrated in Fig.~\ref{fig:PostSelectingAngles}. 
\end{proof}

\begin{figure}
\begin{center}
\includegraphics[width=0.8\linewidth]{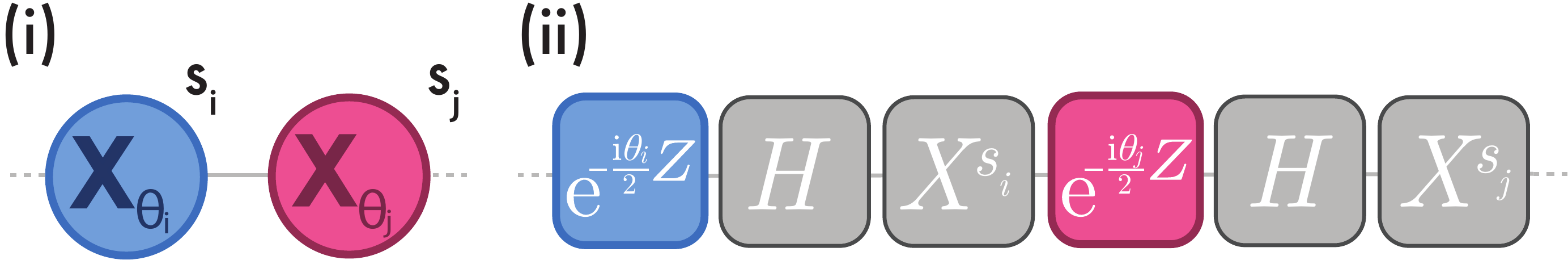}
\end{center}
\caption{(i) Measurement of  $X_{\theta_i},X_{\theta_j}$ on an edge of a 1D cluster state:  the outcomes $s_i,s_j\in\{0,1\}$ are uniformly random. (ii) The associated logical circuit: the byproduct operator $X^{s_i}$ can be propagated forward in circuit  by substituting $\theta_j$ with $(-1)^{s_1}\theta_j$. The argument extends to the full cluster by induction, choosing the 1st qubit to be measured in the $X$ basis (this fixes the input of the logical circuit and does not change its universality properties). The 2D cluster-state case is analogous~\cite{RaussendorfPhysRevLett.86.5188,RaussendorfPhysRevA.68.022312}.}
\label{fig:PostSelectingAngles}
\end{figure}

\subsection{\textsf{\#P}-hardness of approximating output probabilities (proof of Lemma \,\ref{lemma:SharpPHardness})}\label{E}

In this section we  prove Lemma \ref{lemma:SharpPHardness}. Our proof below shows that the ability to approximate the given Ising partition function can be used to approximate the output probabilities of the ``dense'' 2-local long-range IQP circuits of Ref.\  \cite{BremnerPRL117.080501}. The proof further exploits a new technique (Lemma\mot\ref{lemma:1D_IQP}) to implement $O(n^2)$-size long-range IQP circuits in $O(n)$-depth in a 1D nearest-neighbor architecture, which is asymptotically optimal. We regard Lemma\mot\ref{lemma:1D_IQP} of independent interest since the latter dense IQP circuits were argued in Ref.\  \cite{BremnerPRL117.080501} to exhibit  a quantum speedup but, to our best knowledge, linear-depth 1D implementations for them were not previously known. On the other hand, recently, it has been shown that a ``sparse'' sub-family of the latter IQP circuits can be implemented in depth $O(n \mathrm{log} n)$ in a 2D nearest-neighbor architecture \cite{SparseNoisySupremacy}.

\subsubsection*{A 1D linear-depth implementation of dense IQP circuits} We first derive our intermediate result for IQP circuits.
For any positive $n$, we let $\mathcal{C}$ be any  ``dense'' random  $n$-qubit IQP circuit whose gates are uniformly chosen from the set
\begin{eqnarray}\label{app:eq:IQPGateSet}
&\biggl\{\euler^{\imun \theta_{i,j} X_iX_j} , \euler^{\imun \theta_i X_i}\: : \: i,j\in\{1,\ldots,n\},\,\\
&\theta_i,\theta_{i,j} \in\left\{\tfrac{\uppi k}{8},k=0,\ldots,7\right\}\biggr\},
\end{eqnarray} 
which contains arbitrary long-range interactions in a fully-connected architecture.
\begin{lemma}[Dense IQP circuits]\label{lemma:1D_IQP}
Dense $n$-qubit IQP circuits of\mot{}(\ref{app:eq:IQPGateSet})  gates can be implemented in $\Theta(n)$-depth in a 1D nearest-neighbor architecture. 
\end{lemma}
\begin{proof}
It is easy to see that $n^2$-size 2-local quantum circuits require $\Omega(n)$ depth to be implemented. Our proof gives a matching upper bound for the given IQP circuits.

Recall that IQP gates can performed in any order (as they commute). Hence, by reordering gates and redefining the $\theta_{i,j}$  angles, any given IQP circuit $\mathcal{C}$ can be put in a normal form $\mathcal{C}'$ that contains at most one single-qubit gate per qubit and one two-qubit gate per pair of qubits.  Our approach now is to introduce additional layers of nearest-neighbor SWAP gates following layers of two-qubit gates (Fig.\ref{fig:LinearDepthIQP}).
\begin{figure}
\centering
\includegraphics[width = 0.8\linewidth]{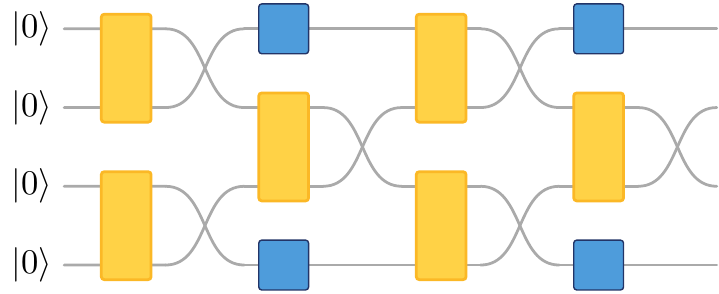}
\caption{Linear-depth implementation of dense IQP circuits (\ref{app:eq:IQPGateSet}). We illustrate our algorithm for $4$-qubits. Yellow- (resp.\ blue-) blocks implement the two- (resp.\ one-) qubit gates in (\ref{app:eq:IQPGateSet}). The  contains 4 single-qubit gates (resp.\ 6 two-qubit ones), which coincides with the number of vertices (resp.\ edges) of the complete graph $K_4$.}\label{fig:LinearDepthIQP}
\end{figure}
To illustrate the algorithm, we regard qubits as  ``particles'' moving up or down the line by the action of the SWAP gates. At a given step $t$, we apply a two-qubit IQP gate followed by a SWAP to all pairs of form $(2i-1,2i)$ for $1\leq i \leq  \lfloor n/2\rfloor$ when $t$ is even (resp.\ $(2i,2i+1)$ for $1\leq i \leq   \lfloor (n-1)/2\lfloor$ when $t$ is odd). By iterating this process $n$-times,  the qubit initially in the $i$th position in the line (with arbitrary $i$) travels to the $n-i+1$th position, meeting every other qubit exactly one time along the way due to the Intermediate Value Theorem; each 2-qubit gate of $\mathcal{C}'$ is implemented in one of these crossing. Furthermore, each qubit spends one step without meeting any qubit when they reach the line's boundary; at these points, single-qubit gates can be implemented.
\end{proof}

\subsubsection*{Proof of Lemma \ref{lemma:SharpPHardness}}

Below, we denote with $\Gamma:=\{\beta \in\{0,\theta\}^{mn}: p_\beta\neq 0\}$ the set of allowed configurations for $\beta$ in step \ref{Exp:Preparation}; $x\in\{0,1\}^n$ (resp.\ $y\in\{0,1\}^{N-n}, N=\mu m n$) be the measurement outcomes of the $n$ right-most primitive qubits (resp.\ remaining ones) after step \ref{Exp:Measurement}; and  $q(x,y,\beta)$ be the final total probability of observing  the values $x,y,\beta$. 

As in section \ref{D}, it will be convenient to recast our architectures as non-adaptive using mappings \ref{app:MapMBQCA1A3}-\ref{app:MapMBQCA2}. In this picture,  the following identity  readily follows from  standard properties of $X$-teleportation circuits \cite{Zhou00LogicGateConstruction,Childs05Unified_MBQC_PRA}:
\begin{align}\label{app:eq:MarginalProbs}
q(x,y,\beta)= & q(x, y| \beta) p_\beta=q(x | y,\beta)\frac{1}{2^{N-n}}\frac{1}{|\Gamma|}, \quad \\
&\textnormal{ for any } x\in\{0,1\}^n, y\in\{0,1\}^{N-n}, \beta \in\Gamma.\notag
\end{align}
Above, we used that\black{} $p_\beta$ is uniformly supported over $\Gamma$ (by design) as well as $q(y|\beta)=1/2^{N-n}$, which follows from  standard properties of $X$-teleportation circuits \cite{Zhou00LogicGateConstruction,Childs05Unified_MBQC_PRA}.
Note that  $\textnormal{prob}(a,b|\beta)$ in  (\ref{eq:Probabilities=PartFunct}) and  $q(x,y|\beta)$ as above are identical  probability distributions up to  a relabeling $(x,y)=\ell(a,b)$ of the random variables. Thus,  if $\widetilde{q}(\ell(a,b)|\beta)$ approximates $q(\ell(a,b)|\beta)$ up to relative error~$1/4+o(1)$, then, $|\widetilde{\mathcal{Z}}^{( \uppi a,\frac{\uppi}{4} b + \beta)}|^2:=\widetilde{q}(\ell(a,b)|\beta)2^{N_X+N_Z/2}$ approximates $|\mathcal{Z}^{( \uppi a,\frac{\uppi}{4} b + \beta)}|^2$ with the same error. Hence, the proof reduces to showing that approximating $q(x,y|\beta)$ for architectures \ref{enum:Arc1}-\ref{enum:Arc3}  is \textsf{\#P}-hard for the given error and $m\in O(n^2)$.

Next, recall that the output probabilities $a_{z}=|\bra{z_1,\ldots,z_k}\mathcal{C}|0\rangle^{\otimes k}|^2$ of an arbitrary $k$-qubit dense IQP circuits $\mathcal{C}$ as in (\ref{app:eq:IQPGateSet})  are \textnormal{\textsf{\#P}}-hard to approximate up to relative error $1/4+o(1)$ \cite{Fujii13_CommutingCircuitsIsing,Goldberg14_ComplexityIsingTutte}. Via Lemmas \ref{lemma:PostSelection} and ~\ref{lemma:1D_IQP}, the latter can further be implemented in  our architectures using lattices with $n\times O(n^2)$ qubits for some $n:= k + r$ with $r\in O(k)$; to apply Lemma \ref{lemma:PostSelection}, we can either decompose $\mathcal{C}$ exactly as a 1D homogeneous Clifford+$T$ circuits \cite{Soeken13RootsPauliMatrices}, or use the gadgets in the proofs of Lemmas~\ref{lemma:lemma:MBQCX-YOnColumns}-\ref{lemma:MBQCX-YOnSites} to directly implement the IQP gates. Now, let $\ket{\psi}_{y,\beta}$ denote the state vector of the $n$ right-most primitive qubits after observing $y,\beta$.  It follows from our discussion that $\ket{\psi}_{y,\beta}=\left(\mathcal{C}\ket{0}^{\otimes k}\right)\ket{0}^{\otimes r}$ for some efficiently-computable value of $y,\beta$.
Defining $\overline{z}:=(z_1,\ldots,z_k,0_{k+1},\ldots,0_{k+r})$, it follows that $a_x=q(\overline{z}|y,\beta)$. Further, if $\widetilde{q}(\overline{z},y|\beta)$ approximates $q(\overline{z},y|\beta)$ up to relative error ~$\eta>0$, then, $\widetilde{q}(\overline{z}|y,\beta):=\widetilde{q}(\overline{z},y|\beta)2^{N-n}$  approximates $q(\overline{z}|y,\beta)$ with the same error. Hence, approximating $q(\overline{z}|y,\beta)$ up to relative error $1/4+o(1)$ is \textnormal{\textsf{\#P}}-hard.\hfill\qedsymbol

\subsection{Hardness argument (proof of Theorem\,�\ref{thm:Main})}\label{F}

Finally, we prove Theorem \ref{thm:Main}. 
Similarly to Refs.\  \cite{AaronsonArkhipov13LinearOptics4,BremnerPRL117.080501}, we apply Stockmeyer's Theorem~\cite{Stockmeyer85ApproxiationSharpP} to relate the problems of approximately sampling from output distributions of quantum circuits to approximating individual output probabilities. 
Our proof is by contradiction: assuming that the worst-case \textsf{\#P}-hardness of estimating the partition functions of the Ising models (\ref{eq:RandomIsingModel}) extends to average-case (Conjecture \ref{conj:IsingAverageComplexity}) and that the  output probabilities of architectures \ref{enum:Arc1}-\ref{enum:Arc3} are sufficiently anti-concentrated (Conjecture \ref{conj:AntiConcentration}), we show that the existence of a classical algorithm for sampling from the latter within constant $\ell_1$-norm implies that an $\textsf{FBPP}^\textsf{NP}$ algorithm can solve \textsf{\#P}-hard problems; this leads to a collapse of the polynomial hierarchy to its third level in contradiction with  Conjecture
\ref{conj:ConjecturePolyHiear}.

Let $\Gamma:=\{\beta \in\{0,\theta\}^{mn}: p_\beta\neq 0\}$ be the set of allowed  $\beta$ configurations in step~\ref{Exp:Preparation}; $x\in\{0,1\}^n$ (resp.\ $y\in\{0,1\}^{N-n}, N=\mu m n$) be the measurement outcomes of the $n$ right-most primitive qubits (resp.\ remaining ones) after step \ref{Exp:Measurement}; and  $q(x,y,\beta)$ be the final total probability of observing  the values $x,y,\beta$. As a preliminary, we prove that if Conjecture~\ref{conj:AntiConcentration} holds, then the  probability distribution $q(x,y,\beta)$ associated to  random-input states and measurement outcomes of architectures \ref{enum:Arc1}-\ref{enum:Arc3} is anti-concentrated. We first note that $q(x|y,\beta)$ in (\ref{app:eq:MarginalProbs}) coincides with the output distribution of some $n$-qubit $O(m)$-depth circuit $\mathcal{C}_{y,\beta}$ of gates of form\mot{}(\ref{eq:TranslationInvariantGateset}): this is easily seen from mappings \ref{app:MapMBQCA1A3}-\ref{app:MapMBQCA2} and standard properties of $X$-teleportation \cite{RaussendorfPhysRevA.68.022312,Zhou00LogicGateConstruction,Childs05Unified_MBQC_PRA} (cf.\ also the next section and Fig.~\ref{fig:circuits}). For arbitrary $y\in\{0,1\}^{N-n}$, $\beta\in\Gamma$, let us now define 
\begin{align}\label{app:GammaAlpha}
\gamma_{y,\beta}&:=\frac{\left|\{x\in\{0,1\}^{n}:q(x|y,\beta)\geq 1/2^{n}\}\right|}{2^n},\\
\alpha&:=\frac{\left|\left\lbrace y\in\{0,1\}^{N-n},\beta\in\Gamma: \gamma_{y,\beta} \geq 1/e\right\rbrace\right|}{2^{N-n}|\Gamma|};
\end{align}
$\gamma_{y,\beta}$ is the  fraction of $\mathcal{C}_{y,\beta}$'s output probabilities larger than $1/2^n$, and $\alpha$ is the fraction of $\mathcal{C}_{y,\beta}$ circuits that fulfill (\ref{eq:AntiConcentrationProb}). Conjecture~\ref{conj:AntiConcentration}  states that $\gamma_{y,\beta}\geq 1/e$ for $m\in O(n)$. Consequently,  for $n\times O(n)$ lattices, (\ref{app:eq:MarginalProbs}) implies that
\begin{equation}
\label{app:AntiConcentrationTotal}
\mathrm{prob}_{x,y,\beta}\left( q(x,y,\beta)\geq \frac{1}{2^{N}|\Gamma|}\right)\geq 1/e.
\end{equation}Furthermore, since  $q(y,\beta)$  is uniformly distributed over its support, it also follows from (\ref{app:eq:MarginalProbs}) that
\begin{align}
 \mathrm{prob}_{x,y,\beta}\left( q(x,y,\beta)\geq \frac{1}{2^{N}|\Gamma|}\right)&=\sum_{y,\beta}\frac{\gamma_{y,\beta}}{2^{N-n}|\Gamma|} \label{eq:AntiConFullVSMarginal1}\\
& =\mathbb{E}_{y,\beta}\left(\gamma_{y,\beta}\right)\geq\frac{\alpha}{e},\notag
\end{align}
for a uniformly random $x,y\in\{0,1\}^{N},\beta\in\Gamma$. Eq.~(\ref{eq:AntiConFullVSMarginal1}) tells us that the robustness of the anti-concentration inequality\mot{}(\ref{app:AntiConcentrationTotal}) can be tested by computing the average value $\mathbb{E}_{y,\beta}\left(\gamma_{y,\beta}\right)$ of $\gamma_{y,\beta}$, or by estimating the fraction $\alpha$.  As it is discussed further in appendix \ref{G}, $\gamma_{y,\beta} \times e$ and $\alpha$ are expected to converge to one for universal 1D nearest-neighbor circuits as $n$ grows asymptotically \cite{EmersonPseudoRandom,brown2008,Harrow2009,Emerson05ConvergenceRandomCircuits} in the regime $m\in O(n)$ \cite{Kim13BallisticSpread,Hosur2016,Brandao2016LocalRandomQuantumCircuitsAreDesigns}. In appendix \ref{G}, Fig.\ \ref{fig:anticoncentration box}, we present  numerical evidence that $\mathbb{E}_{y,\beta}\left(\gamma_{y,\beta}\right)\rightarrow 1/e$, $\alpha\rightarrow 1$ in the asymptotic limit and a tight agreement for $n\geq 9$; more strongly, we also find that $\gamma_{y,\beta}$ is nearly $1/e$ for almost every uniformly-sampled instance for $n\geq 9$ and that $q(x|y,\beta)$  is Porter-Thomas distributed, which is a signature of Haar-random chaotic unitary processes 
\cite{PorterThomas,HaakeChaos,EmersonPseudoRandom,Emerson05ConvergenceRandomCircuits,brown2008,Boixo161608.00263}.

We are now ready to prove Theorem \ref{thm:Main}. For any architecture\mot{}\ref{enum:Arc1}-\ref{enum:Arc3},  we let $a$ denote an element of $\{0,1\}^{N}\times \Gamma$, 
and assume that the output distribution $p_{c}(a)$ of a classical BPP algorithm fulfills that
\begin{equation}\label{app:eq:EpsilonFar}
\sum_{a\in \{0,1\}^{N}\times \Gamma} |p_c(a)-q(a)| \leq \varepsilon,
\end{equation}
for a  constant $\varepsilon\geq 0$. By Stockmeyer's Theorem \cite{Stockmeyer85ApproxiationSharpP} and the triangle inequality,  there exists an $\textsf{FBPP}^\textsf{NP}$ algorithm that computes an estimate $\widetilde{p_c(a)}$ such that
\begin{equation}
|\widetilde{p_c(a)}-q(a)|\leq \frac{q(a)}{\ppoly{N}}+|p_c(a)-q(a)| \left(1+\tfrac{1}{\ppoly{N}}\right),\notag
\end{equation}
where we used $\log |\Gamma| \in O(N)$ to remove dependencies on $|\Gamma|$. From Markov's inequality and (\ref{app:eq:EpsilonFar}), we get that 
\begin{equation}
\textnormal{prob$_{a}$}\left(|p_c(a)-q(a)|\geq\frac{\varepsilon}{2^N|\Gamma|\delta}\right)\leq \delta,
\end{equation}
for any constant $0<\delta<1$, where  $a\in\{0,1\}^{N}\times \Gamma$ is picked  uniformly at random.
Hence,
\begin{align}\label{app:eq:MarkovApproximation}
|\widetilde{p_c(a)}-q(a)|\leq \frac{q(a)}{\ppoly{N}}+\frac{\varepsilon\left(1+o(1)\right)}{\delta 2^N|\Gamma|}
\end{align} 
with probability at least $1-\delta$ over the choice of $a$.

We now claim that Eqs.\ (\ref{app:eq:MarkovApproximation}) and (\ref{app:AntiConcentrationTotal})  simultaneously  hold for a single $q(a)$ with probability  $(1/e)\cdot(1-\delta)$, since  a classical description of Ref.\ \ref{enum:Arc1}-\ref{enum:Arc3} does not reveal to a classical sampler which output probabilities are \textsf{\#P}-hard to approximate:  hence, the latter cannot  adversarially corrupt the latter. This is manifestly seen at the encoded random circuit level, due to the presence of random byproduct operators of form $\prod_{i}X_i^{y_i}$ (with random $y_i$), which obfuscate the location of the \textsf{\#P}-hard probabilities from the sampler \cite{BremnerPRL117.080501,Bremner17SupremacyReview}. Hence, setting $\varepsilon=\gamma/8,\delta=\gamma/2, \gamma=1/e$, we obtain that
\begin{equation}
|\widetilde{p_c(a)}-q(a)|\leq\left(\frac{1}{4}+o(1)\right) q(a)
\end{equation}
with probability at least $\gamma(1-\gamma/2)>0.3$ over the choice  of $a$. Setting $\varepsilon = 1/22 < \gamma/8$, the above procedure yields an approximation $\widetilde{p_c(a)}$ of  $q(a)$ up to relative error $1/4+o(1)$. Using (\ref{eq:Probabilities=PartFunct}) and (\ref{app:eq:MarginalProbs}), we obtain an $\textsf{FBPP}^\textsf{NP}$ algorithm that approximates $|\mathcal{Z}^{\alpha,\beta}|^2$ with relative error  $1/4+o(1)$  for at least a 0.3 fraction of the instances. This yields a contradiction.\black{}

As final remarks, note that the above argument is  robust to small finite-size variations to the threshold $\gamma=1/e$ in Conjecture\mot\ref{conj:AntiConcentration}, Eq.\mot(\ref{eq:AntiConcentrationProb}) since the constants $\varepsilon=\gamma/5, \delta=\gamma/2$, and $\gamma(1-\gamma/2)$ have only linear and quadratic dependencies on  $\gamma$. 
Also, notice that Conjectures \ref{conj:IsingAverageComplexity} and \ref{conj:AntiConcentration} enter the above argument in order to allow for a constant additive error $\varepsilon$, which is key for a real-life demonstration of a quantum speedup. \je{Additive errors give rise to demanding, but not unrealistic prescriptions.}
However, in the ideal case where one assumes no-sampling errors, or  multiplicative sampling errors, our result holds even without these conjectures 
via the arguments in
Refs.\  \cite{terhal2004adaptive,BremnerJozsaShepherd08}.\black{}

\section{Conclusion}

In this work, we have established feasible and simple schemes for quantum simulation that  exhibit a superpolynomial quantum speedup with high evidence, in a
complexity-theoretic sense. As such, this work is expected to significantly contribute to 
bringing notions of quantum devices outperforming classical supercomputers closer to reality.
This work can be seen as an invitation towards a number of further exciting research directions: 
While the schemes presented may not quite yet constitute experimentally realizable blue-prints, it should be clear that steps already experimentally taken are very similar to those discussed. It seems hence interesting to explore detailed settings for cold atoms or trapped ions in detail, requiring little local control and allowing for comparably short coherence times. 
What is more, it appears obvious that further complexity-theoretic results on intermediate problems seem needed to fully capture the potential of quantum devices outperforming classical computers without being universal quantum computers. It is the hope that the present work can contribute to motivating such further work,
guiding experiments in the near future.

%

\section{Acknowledgements}
 e
We thank Scott Aaronson, Dan E.\ Browne, Andreas Elben, Bill Fefferman, Wolfgang Lechner and Ashley Montanaro for discussions. JBV thanks Vadym Kliuchnikov and Neil Julien Ross for helpful comments. This work was supported by the EU (AQUS), the Templeton and Alexander-von-Humboldt Humboldt Foundations, The ERC (TAQ), and the DFG (CRC 183, EI 519/7-1). RR is funded by NSERC, Cifar, IARPA and is Fellow of the Cifar Quantum Information Program.

\phantomsection


\begin{thebibliography}{95}
\bibitem [{Com()}]{CommutingComment}%
  \BibitemOpen
  \href@noop {} {}\bibinfo {note} {This also implies that not the entire
  resource has to necessarily be prepared in a single preparation, but can be
  sequentially prepared.}\BibitemShut {Stop}%
\bibitem [{Tra()}]{Trapped}%
  \BibitemOpen
  \href@noop {} {}\bibinfo {note} {On the one hand, configurable lattice
  systems of two-dimensional arrays of individually trapped ions can be
  conceived \cite{2DIons}. Probably more feasible, on the other hand, are
  flattened architectures, in which the 2D setting is traded for suitable
  long-ranged quantum gates.}\BibitemShut {Stop}%
\end{thebibliography}

%

\pagebreak
\newpage

\appendix

\begin{widetext}

\section{Full Hamiltonian (\ref{eq:HardIsingModels})}\label{A}

For any  architecture \ref{enum:Arc1}-\ref{enum:Arc3} (Fig.~\ref{fig:Experiment}), let $\mathcal{L}_\mathrm{P}=(V_\mathrm{P},E_\mathrm{P})$ be the sublattice of $\mathcal{L}$ containing all primitive qubits (i.e., the square lattice subgraph of $\mathcal{L}$). Further,  let  $\mathcal{L}_{\mathrm{DB}}=(V_\mathrm{DB},E_\mathrm{DB})$ be $\mathcal{L}$'s sublattice containing all dangling-bonds for architecture \ref{enum:Arc3} and the empty graph otherwise. For any $i\in V$, let $\mathrm{deg}_{\mathrm{P}}(i)$ (resp.\ $\mathrm{deg}_{\mathrm{DB}}(i)$) be the number of primitive (resp.\ dangling-bond) qubits connected to $i$ in $\mathcal{L}$. Then, the full Hamiltonian~(\ref{eq:HardIsingModels}) of the experiment reads
\begin{align}
H =\underbrace{\sum_{(i,j)\in E_{\mathrm{P}}} \frac{\uppi}{4} Z_i Z_j - \sum_{i'\in V_\mathrm{P}} \frac{\uppi}{4} \mathrm{deg}_\mathrm{P}(i')Z_{i'}}_{H_{\mathrm{CZ}}}+
\underbrace{\sum_{(k,l)\in E_{\mathrm{DB}}} \frac{\uppi}{16} Z_k Z_l -  \sum_{k'\in V_\mathrm{DB}} \frac{\uppi}{16} \mathrm{deg}_\mathrm{DB}(k')Z_{k'}}_{H_{\mathrm{CT}}}.\label{eq:FullHamiltonian}
\end{align}
Above, the Hamiltonian $H_{\mathrm{CZ}}$ (resp.\ $H_{\mathrm{CT}}$) implements a $CZ$ (resp.\ $CT$) gate on every edge of the bright (resp.\ dark) sublattice in Fig.~\ref{fig:Experiment}. Note that $H_{\mathrm{CT}}$ is not present in architectures \ref{enum:Arc1}-\ref{enum:Arc2}. Also, realize that $\mathrm{deg}_{\mathrm{P}}(i)$ takes value 4 on bulk qubits, 2 on the corners and 3 elsewhere on edges; for architecture \ref{enum:Arc3}, $\mathrm{deg}_{\mathrm{DB}}(i)$ takes value 1 everywhere.

\section{Mapping output probabilities to Ising partition functions}\label{B}

Let $\mathcal{L}_X=(V_X,E_X)$ (resp.\ $\mathcal{L}_Z=(V_Z,\emptyset)$) be the primitive-qubit square sublattice of $\mathcal{L}$ (resp.\ the disjoint union of all dangling-bond qubits), and pick $\mathcal{L}_X$ to be the lattice $\mathcal{L}_\mathrm{sq}$ in section \ref{main result}, Eq.\mot{}\ref{eq:RandomIsingModel}.  Let $\alpha:=\uppi a$, $\vartheta:=\beta+\frac{\uppi}{4}b$, where we let $b$ be the string of outcomes of the  $Z$ measurements  in architecture \ref{enum:Arc3}, and  define $b_i:=0,i=1,\ldots,N_Z$, by convention,  for architectures \ref{enum:Arc1}-\ref{enum:Arc2}. Further, let $\theta=\uppi/8$ for architecture \ref{enum:Arc2} and  $\theta=\uppi/4$ otherwise. We now prove equation (\ref{eq:Probabilities=PartFunct}) using formula~ (\ref{eq:FullHamiltonian}):
\begin{align}
|\bra{a,b} \euler^{-\imun H} \ket{\psi_\beta}|&
=|\bra{a,b} \euler^{-\imun (H + \sum_{i\in V} \frac{\beta_i}{2} Z_i) } \ket{+}^{\otimes N}|
=|\bra{a, b} \euler^{-\imun (\sum_{(i,j)\in E} J_{i,j}Z_iZ_j - \sum_{i\in V} (h_i-\frac{\beta_i}{2}) Z_i) } \ket{+}^{\otimes N}|\notag\\
&=\frac{1}{\sqrt{2}^{N_Z}}|\bra{a} \euler^{-\imun (\sum_{(i,j)\in E_X} \frac{\uppi}{4}Z_iZ_j -\sum_{i\in V_X}\sum_{(i,k)\in E_Z} (h_i-\frac{\beta_i}{2}- \frac{\uppi b_i}{8}) Z_i) } \ket{+}^{\otimes N_X}|\notag
\\
&=\frac{1}{\sqrt{2}^{N_Z}}|\bra{+} \euler^{-\imun\sum_{i\in V_X} \frac{\alpha_i}{2} Z_i}\euler^{-\imun (\sum_{(i,j)\in E_X} \frac{\uppi}{4}Z_iZ_j - \sum_{i\in V_X}\sum_{(i,k)\in E_Z} (h_i-\frac{\vartheta_i}{2}) Z_i) } \ket{+}^{\otimes N_X}|
\notag\\
&=\frac{1}{\sqrt{2}^{N_Z}}|(\bra{0}H)^{\otimes N_X}\euler^{-i H^{(\alpha,\beta)} }(H\ket{0})^{\otimes N_X}|=\frac{1}{\sqrt{2}^{N_Z}2^{N_X}}\left|\sum_{x,y\in\{0,1\}^{N_X}} \bra{x}\euler^{-i H^{(\alpha,\beta)} } \ket{y}\right|\notag\\
&=\left|\frac{\textnormal{tr}\left(\euler^{-\imun H^{(\alpha,\beta)}}\right)}{\sqrt{2}^{N_Z+2N_X}}\right|=
\frac{|\mathcal{Z}^{(\alpha,\beta)}|}{\sqrt{2}^{N_Z+2N_X}},\notag
\end{align}
where we have defined $\ket{+}:=(\ket{0}+\ket{1})/\sqrt{2}$ in the second step, and used that $H^{(\alpha,\beta)}$ is diagonal in the final one. We obtain (\ref{eq:Probabilities=PartFunct}) by squaring; therein, $\alpha_i\in\{0,\uppi\}$, $\vartheta\in\{0,\theta\}$ follows from the definition.

\section{Numerical evidence for anti-concentration of the output distribution}\label{app:numerical evidence}\label{G}
\vspace{-5pt}

In this appendix we present numerical evidence for the validity of Conjecture~\ref{conj:AntiConcentration}, which is exploited in the proof of Theorem~\ref{thm:Main} as discussed in section \ref{F}. Therein, we discussed that for any architectures \ref{enum:Arc1}-\ref{enum:Arc3}, given an initial $n$-row, $m$-column square lattice, there exists an $n$-qubit $D$-depth 
circuit family $\{\mathcal{C}_{\beta,y}\}_{y,\beta}$ of gates of form (\ref{eq:TranslationInvariantGateset}), with $D\in O(m)$, such that $q(x|y,\beta)=|\langle x | \mathcal{C}_{\beta,y}\ket{0}|^2$. If  the circuits $\{\mathcal{C}_{\beta,y}\}_{y,\beta}$ exhibit anti-concentration as in Conjecture~\ref{conj:AntiConcentration}, then $q(x,y,\beta)$ is anti-concentrated as in Eq.\mot(\ref{app:AntiConcentrationTotal}). This is used in section \ref{F} to turn an approximate classical sampler into an $\textsf{FBPP}^\textsf{NP}$ algorithm to approximate single output probabilities with high-accuracy. 

The concrete circuit families associated to  each architecture are derived below and depicted in  Fig.~\ref{fig:circuits}. 
To numerically test Conjecture~\ref{conj:AntiConcentration}, we have performed simulations of randomly generated circuits of gates of form (\ref{eq:TranslationInvariantGateset}) (for each circuit family) in LIQUiD  \cite{liquid} with up to 20 logical qubits. For each system size, we generated 100 random instances for circuits associated to $n\times n$ and  $n \times n^2$ lattices. For each instance, we evaluated exactly the fraction $\gamma_{y,\beta}$ (Eq.\ \ref{app:GammaAlpha}) of output-probabilities fulfilling\mot(\ref{eq:AntiConcentrationProb}). Our results are summarized in  Fig.\ \ref{fig:anticoncentration box}: therein, one can see that both for circuits associated to $n \times n$ and $n \times n^2$ lattices, this fraction quickly
approaches a constant $\gamma = 1/e$ with rapidly decreasing variance with
respect to the choice of circuits. We can conclude that, with very high probability, in a realization of the
proposed experiment the amplitude of the final state of the computation anti-concentrates

As discussed in Refs.\  \cite{BremnerPRL117.080501,SparseNoisySupremacy,Bremner17SupremacyReview}, it might seem a priori counter-intuitive that constant-depth nearest-neighbor architectures  anti-concentrate. However, the above connections between our architectures and random circuits shed key insights into why this behavior is actually natural. As shown in section \ref{D}, the random logical circuits of gates of form (\ref{eq:TranslationInvariantGateset}) encoded in our architectures  are universal for quantum computation. Universal random quantum circuits of increasing depth are known to approximate the Haar measure under various settings \cite{EmersonPseudoRandom,Emerson05ConvergenceRandomCircuits,brown2008,Harrow2009,Brandao2016LocalRandomQuantumCircuitsAreDesigns}. For 1D nearest-neighbor layouts, the latter are expected to reach a chaotic Porter-Thomas-distributed regime \cite{PorterThomas,HaakeChaos,EmersonPseudoRandom} in depth $D\in O(n)$ \cite{Kim13BallisticSpread,Hosur2016,Brandao2016LocalRandomQuantumCircuitsAreDesigns} (cf. \cite{Boixo161608.00263} for further discussion). As an additional piece of supporting evidence for anti-concentration, 
we numerically confirmed that our output-probabilities are close to being Porter-Thomas-distributed in $\ell_1$-norm (Fig.\ \ref{fig:tv distance pt}). Furthermore, our numerics are in agreement  with prior numerical works on MBQC settings \cite{brown2008} and other gate sets in 2D layouts \cite{Boixo161608.00263,Aaronson16_FoundationsQuantumSupremacy}.  

\emph{Circuit families.} For the sake of completeness, we spell out the logical circuits that are effectively implemented in architectures\mot\ref{enum:Arc1}-\ref{enum:Arc2}. The latter are derived  via mappings \ref{app:MapMBQCA1A3}-\ref{app:MapMBQCA2} and $X$-teleportation properties \cite{Zhou00LogicGateConstruction,Childs05Unified_MBQC_PRA}. Examples for $4 \times 2$ lattices are depicted in Fig.~\ref{fig:circuits}.
\begin{figure} 
    \includegraphics[width=.45\textwidth]{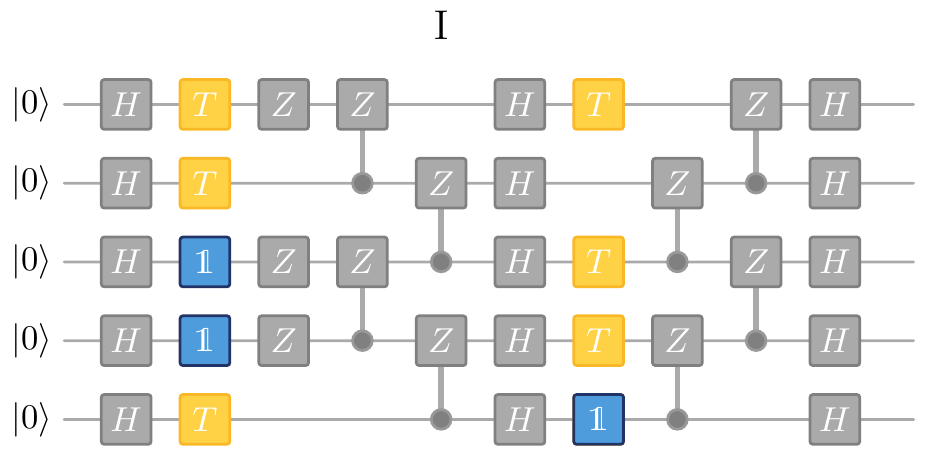}
    \qquad\qquad 
    \includegraphics[width=.45\textwidth]{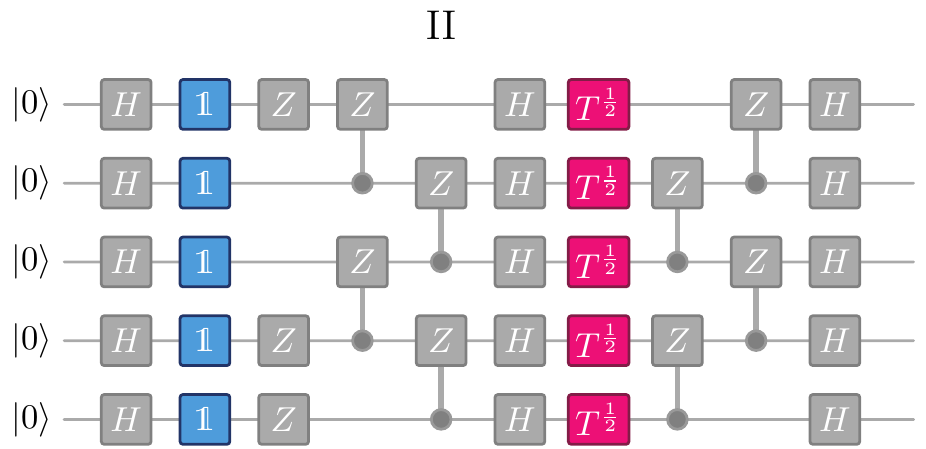}
    \caption{We show the logical circuit corresponding to architectures \ref{enum:Arc1} and
    \ref{enum:Arc3} (left) and \ref{enum:Arc2} (right) for a 5-row 2-column lattices prepared with
    uniformly random $\beta_{i} \in \{ 0, \pi/4\}$ (blue, yellow) (\ref{enum:Arc1} and \ref{enum:Arc3}) and  column-random $\beta_{i} \in \{ 0 , \pi/8 \}$ (blue, crimson) (\ref{enum:Arc2}). 
    At those qubits that have been prepared with $\beta_{i} = 0 $ an identity gate is applied, on those with $\beta_{i} = \pi/4$ a $T$ gate and on those  with $\beta_i = \pi/8$ a $\sqrt{T}$ gate.\label{fig:circuits}}
\end{figure} 
The logical circuit family corresponding to \ref{enum:Arc1} and \ref{enum:Arc3} (resp.\ to \ref{enum:Arc2})  is denoted $\mathcal{F}_{\text{DO}}$ (resp.\ $\mathcal{F}_{\text{col}}$). Let us now label primitive-lattice sites by row-column coordinates $[i,j]$. The circuits are generated inductively, starting from the left column $j=1$. Measurements are ordered from left to right. The computation begins on  the $\ket{+}^{\otimes n}$ state and proceeds as follows:
\begin{enumerate*}
\item Apply the gate $\exp{(\imun \beta_{[i,j]} Z_{[i,j]})}$ to qubit ${[i,j]}$, with $\beta_{[i,j]}$ chosen as in step \ref{Exp:Preparation} for \ref{enum:Arc1}-\ref{enum:Arc2}; for \ref{enum:Arc3}, we let $\beta_{[i,j]}:=s_{[i,j]}\uppi/4$,  where $s_{[i,j]}$ is the outcome after measuring the dangling-neighbor of $[i,j]$.
\item  If $j < m$,  apply a random  $Z_{[i,j]}^{a_{[i,j]}}$ gate to  every qubit $[i,j]$, where   $a_{[i,j]}$ is the outcome of the 
measurement  at site $[i,j]$.
\item Apply $CZ$ on all neighboring qubits.
\item Apply  a Hadamard gate to each qubit. 
\item If $j= m$, measure in the standard basis and terminate; otherwise increase $j:=j+1$.
\end{enumerate*}
\begin{figure*}[t] 
    \centering
    \includegraphics[width=0.95\textwidth]{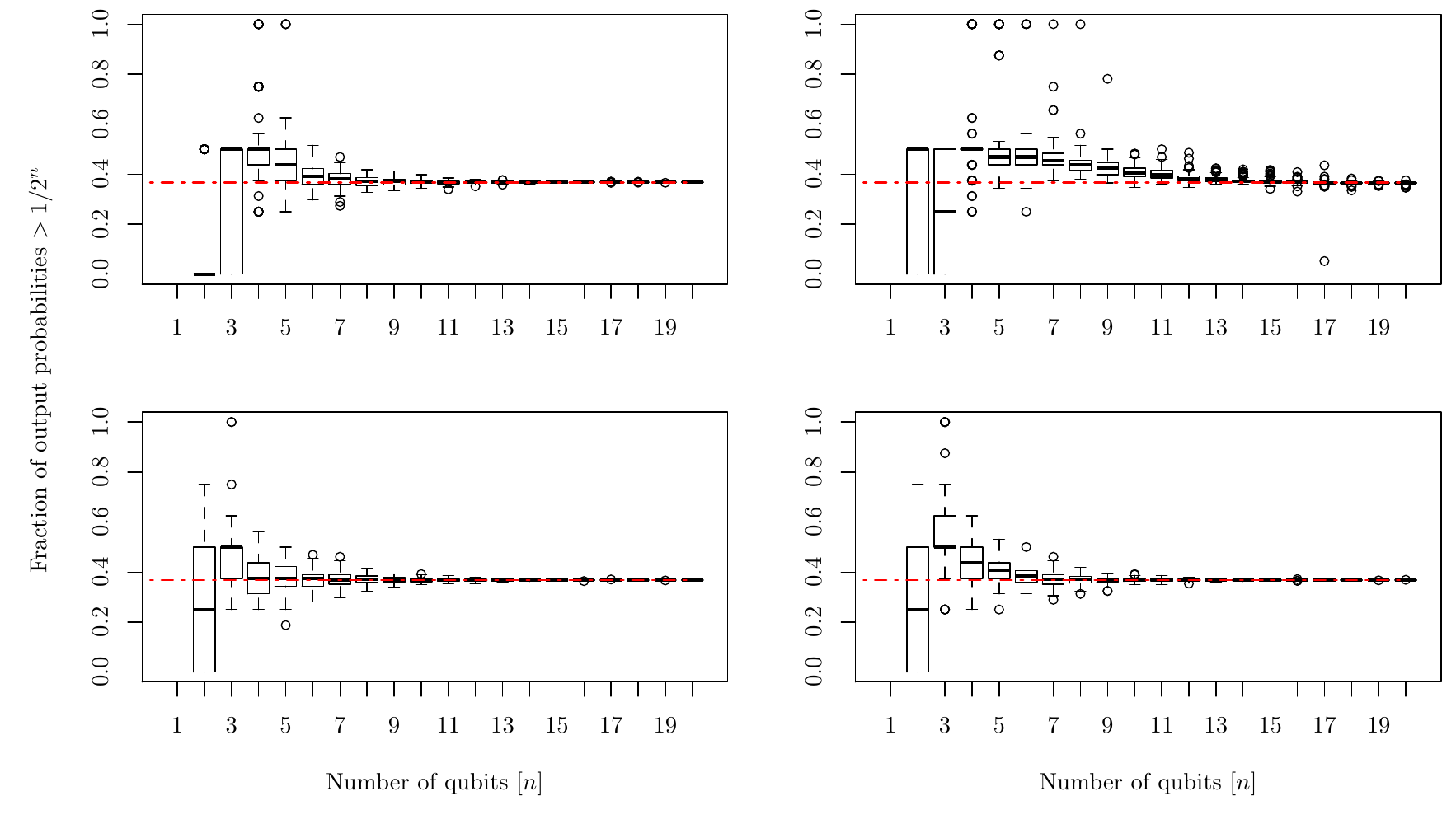}
    \caption{Fraction of output probabilities larger than $1/2^n$ 
    of random circuits drawn from
    the families $\mathcal{F}_{\text{DO}}$ (l.h.s.) and
    $\mathcal{F}_{\text{col}}$ (r.h.s.) for both linear (top) and quadratic
    (bottom) circuit
    depth in the number of qubits the circuit acts upon $n$, 
    i.e., lattices of size $n \times n$ (top) and $n \times n^2$
    (bottom). 
    For each $n$ we draw 100 i.i.d.\ realizations $(\beta,y)$ and thus of the
    circuit $\mathcal{C}_{\beta,y}$ and plot the
    resulting distribution in the form of a box plot. 
    The red dashed line shows the value of $1/e$, which is precisely the value
    to be expected if the output probabilities are Porter-Thomas distributed. 
  \label{fig:anticoncentration box} }
\end{figure*} 

\paragraph{Convergence to the chaotic regime.} 
\begin{figure*}[t] 
    \centering
    \includegraphics[width=0.95\textwidth]{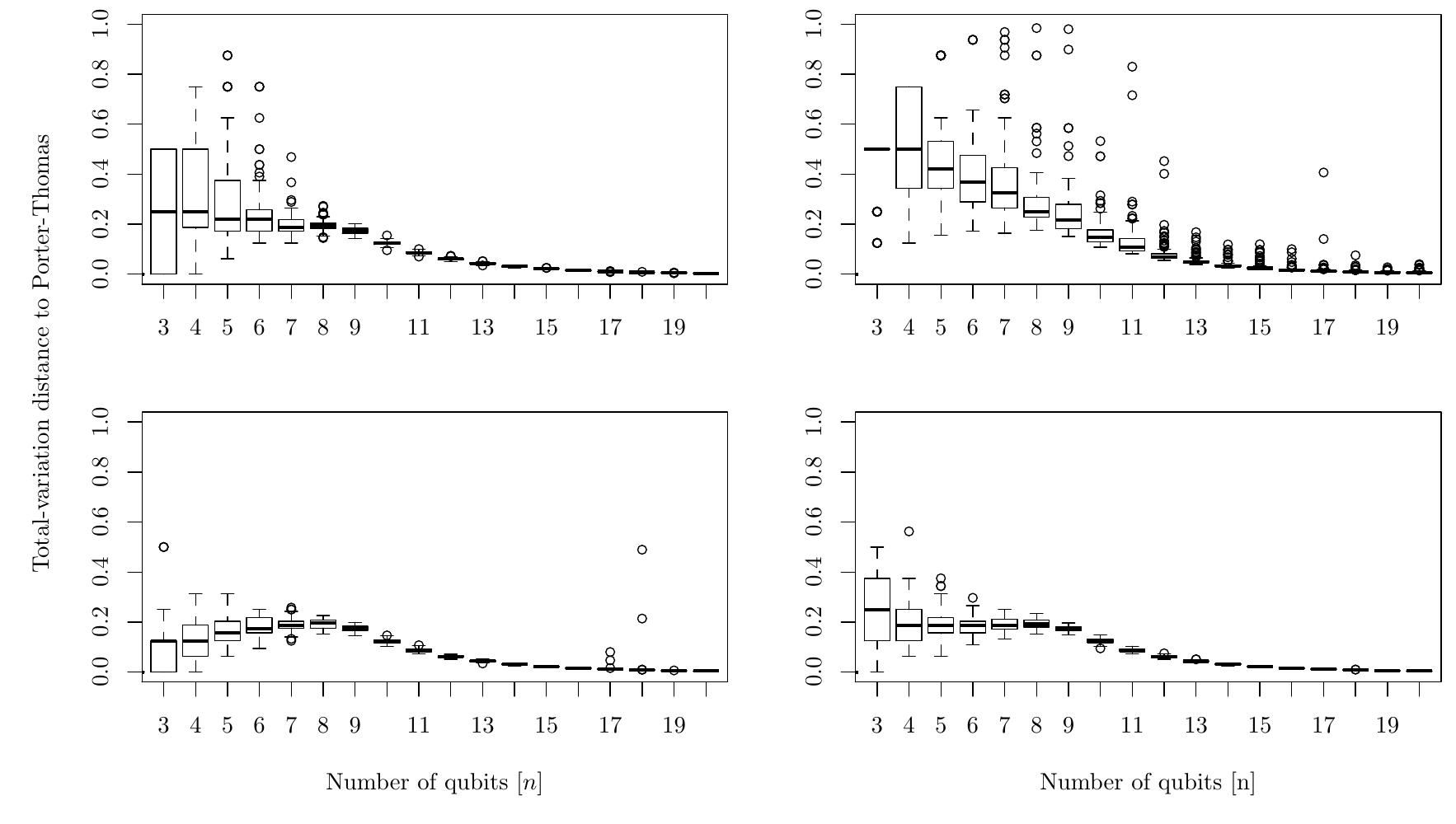}
    \caption{ Total variation distance to the Porter-Thomas distribution of  the empirical distribution of output probabilities of random circuits from
    the families $\mathcal{F}_{\text{DO}}$ (l.h.s.) and
    $\mathcal{F}_{\text{col}}$ (r.h.s.) for both linear (top) and quadratic
    (bottom) circuit
    depth in the number of qubits the circuit acts on $n$, i.e., lattice of size $n \times n$ (top) and $n \times n^2$
    (bottom). 
    For each $n$ we draw 100 i.i.d.\ realizations $(\beta,y)$ and thus of the
    circuit $\mathcal{C}_{\beta,y}$ and plot the
    resulting distribution in the form of a box plot. 
    \label{fig:tv distance pt} }
\end{figure*}  

It is an interesting detail that the value of $\gamma = 1/e$, which we observe
above, is a signature of the exponential distribution (also known as Porter-Thomas distribution) that is known to emerge in chaotic quantum systems for large system sizes \cite{PorterThomas,HaakeChaos,EmersonPseudoRandom,brown2008,Boixo161608.00263,Aaronson16_FoundationsQuantumSupremacy}. 
This distribution is given by 
\begin{equation}
    P_{PT}(p) =  2^n \exp(-2^n p ) \, ,  \label{pt dist}
\end{equation}
and thus anti-concentrates in precisely the fashion observed here.
Note that the same behavior was observed in previous work investigating random MBQC settings \cite{brown2008}, as well as in  recent  work  \cite{Boixo161608.00263},  which investigated random universal circuits on a 2D architecture. Notably, the finite and universal gate sets considered in these works are very similar to the ones considered here. 
Likewise, convergence to the exponential distribution was observed in Ref.\  \cite{Aaronson16_FoundationsQuantumSupremacy} for approximately Haar-random two-qubit unitaries in a 2D setup.

In Fig.~\ref{fig:tv distance pt} we show the total variation distance between
the empirical distributions of output probabilities of the random circuits
generated in our numerical experiments and the discretized Porter-Thomas
distribution. 
We can see that as the number of qubits increases, the
output distributions of random circuits approach Porter-Thomas distribution.

To calculate the total variation distance to the exponential distribution
\eqref{pt dist}, we discretized the interval $[0,1]$ into $m$ bins each of
which contains probability weight $1/m$. In other words, the discretization
$(p_0, p_1, \ldots, p_m)$ is defined by given $p_0 = 0, \, p_m = 1$ and 
\begin{equation}
    \int_{p_i}^{p_{i+1}} P_{PT}(p) \mathrm{d} p = \frac{1}{m} \,. 
\end{equation}
Denote by $Q(p)$ the
numerically observed distribution of output probabilities $p = |\bra{x}
\mathcal{C} \ket{0}|^2$ over the set $\Omega = \{ [p_i,p_{i+1} ]\}_{i =
0,\ldots, m}$. 
The total variation distance between $P$ and the exponential distribution is then given by 
\begin{align}
    \Vert P - Q \Vert_{\mathrm{TV}} = \frac{1}{2} \sum_{X \in \Omega} | P(X) -
    1/m | \, . 
\end{align}
Since the number of samples we obtain in each run is given by $2^n$we choose the number of bins $m$ depending on $n$. Specifically, we choose $m = \mathrm{min} \{ \lceil 2^n/5 \rceil,100 \}$ to allow fair comparison for small $n$.

\section[\#P-hardness from n x O(n) lattices]{\textsf{\#P}-hardness from $n  \times O(n)$ lattices}\label{H}

In this appendix, we show that \textsf{\#P}-hardness of approximating output probabilities in Lemma \ref{lemma:SharpPHardness}  arises already for $n\times m$-qubit lattices  with $m\in O(n)$. Specifically, we prove this by introducing two slight modifications in architectures \ref{enum:Arc1}-\ref{enum:Arc2} (cf.\mot{}Fig.\ \ref{fig:PostSelectingAngles}) without changing the fundamental structure of the basic layout of the steps \ref{Exp:Preparation}-\ref{Exp:Measurement}. 

The key idea is to introduce different types of input states (with different $\beta_i$s) in the preparation step \ref{Exp:Measurement}: specifically, we  pick $\beta_i\in\{0,\theta\}$ for qubits on  odd-row sites, $\beta_i\in\{0,\uppi/2\}$ on even-row ones; additionally, we perform  a local Hadamard rotation on even-row sites before the Ising Hamiltonian evolution in step \ref{Exp:Evolution} begins. This is shown in Figure~\ref{fig:VariationArc1}. The net effect is to initialize even-site qubits on either $\ket{0}$ (the +1-eigenstate of $Z$) or $\ket{-\imun}:=\ket{0}-\imun \ket{1}$ (the -1-eigenstate of $Y$)  at random. Qubits initialized in $\ket{0}$  are invisible to the Ising evolution (\ref{Exp:Evolution}) and effectively become unentangled from the computation. Further, preparing a qubit in $\ket{\mathrm{-}\imun}$, evolving (\ref{Exp:Evolution})  and measuring $X$ is equivalent to  preparing $\ket{+}$ and measuring $Y$ instead at the end of the computation.  Again, we choose the $\ket{\psi_\beta}$ to be fully disordered (DO) for architecture \ref{enum:Arc1}. For architecture \ref{enum:Arc2}, we pick $\ket{\psi_\beta}$ to be TI in one direction with period at most 4, i.e.,  TI$_{(4,\infty)}$-symmetric in our notation.

The full experiment can now be mapped to non-adaptive MBQC analogue to \ref{app:MapMBQCA1A3} with two differences: first, the MBQC acts on a graph state vector $\ket{G}$ \cite{nest06Entanglement_in_Graph_States}, instead of a cluster state, whose underlying graph $G$ is derived from the 2D lattice by deleting $\ket{0}$-state vertices (the output probabilities of the computation can be mapped  to an Ising model on $G$ using the tools of appendices\mot{}A-B); second, the remaining vertices on even columns are measured on the $Y$ basis. We now pick $n=2k-1$  and study the logical-circuit of the MBQC in  two scenarios:
\begin{itemize}
\item[(i)] All even-row qubits are initialized in $\ket{0}$. The MBQC acts on a graph state vector $\ket{G'}$ that is the  product of $k$ disconnected 1D cluster states. Modulo byproduct operators, the local measurements drive a  random logical $k$-qubit circuit of single-qubit gates $\{R_{i}^{a_i}(\theta):=H_i \euler^{\imun (a_i \theta) Z_i}, a_i\in\{0,1\}\}$ and depth $m-1$ (information flows on odd rows). For architecture \ref{enum:Arc2}, the latter circuit inherits a TI$_{(2,\infty)}$-symmetry from the TI$_{(4,\infty)}$-one of the input state vector $\ket{\psi_\beta}$.
\item[(ii)] Even-column even-row  qubits are initialized in $\ket{0}$;  even-column odd-row ones are left unspecified. The MBQC acts on a cluster-state with ``holes'' $\ket{G''}$ as in Refs.\  \cite{Bravyi07_MBQC_toric_code,VanDenNest08_CompletenessClassicalIsingModel_UniversalQC}. Information flows again on odd-rows. If an even-column odd-row qubit $i$ is prepared in $\ket{-\imun}$ and measured in the $X$ basis (or, equivalently, in $\ket{+}$ and measured in the $Y$ basis), we obtain a  reduced graph-state vector $\ket{G'''}$ whose graph $G^{'''}$ is obtained from $G^{''}$ by contracting the edges incident to $i$ \cite{Bravyi07_MBQC_toric_code,VanDenNest08_CompletenessClassicalIsingModel_UniversalQC}. Thus,  
$\ket{-\imun}$ state vectors between the odd qubit lines let us  implement logical entangling gates of form $$E^{b,c}(\theta):=\left(\prod_{i=1}^{k} H_i\right) \left(\prod_{i=1}^{k-1} CZ_{i,i+1}^{b_i}\right)
\left( \prod_{i=1}^{k} \euler^{-\imun (c_i \theta) Z_i}\right),\qquad b_i,c_i,\in\{0,1\},$$ 
where $b_i=1$ if the qubit between lines $2i,(2i-1)$ is in $\ket{-\imun}$ and zero otherwise; and $c_i$ indicates whether we measure $X$ or $X_{-\theta}$ on the $(2i-1)$th line. Postselecting $b_i$  gives us the ability to implement non-translation-invariant 2-qubit entangling gates between qubit lines at will. Again, for architecture \ref{enum:Arc2}, the gate $E^{b,c}(\theta)$ inherits a TI$_{(2,\infty)}$ symmetry.
\end{itemize}
Combining the above facts, it follows that we can  simulate arbitrary  $k$-qubit nearest-neighbor Clifford+$T$ circuits in the modified architecture \ref{enum:Arc1} via postselection, using  lattices with $(2k-1)\times (2k-1)$ qubits. The latter can efficiently implement the \textsf{\#P}-hard IQP circuits of Lemma~\ref{lemma:1D_IQP} with a constant-overhead factor. 

\begin{figure}[t]
\centering 
\includegraphics[width=0.6\textwidth]{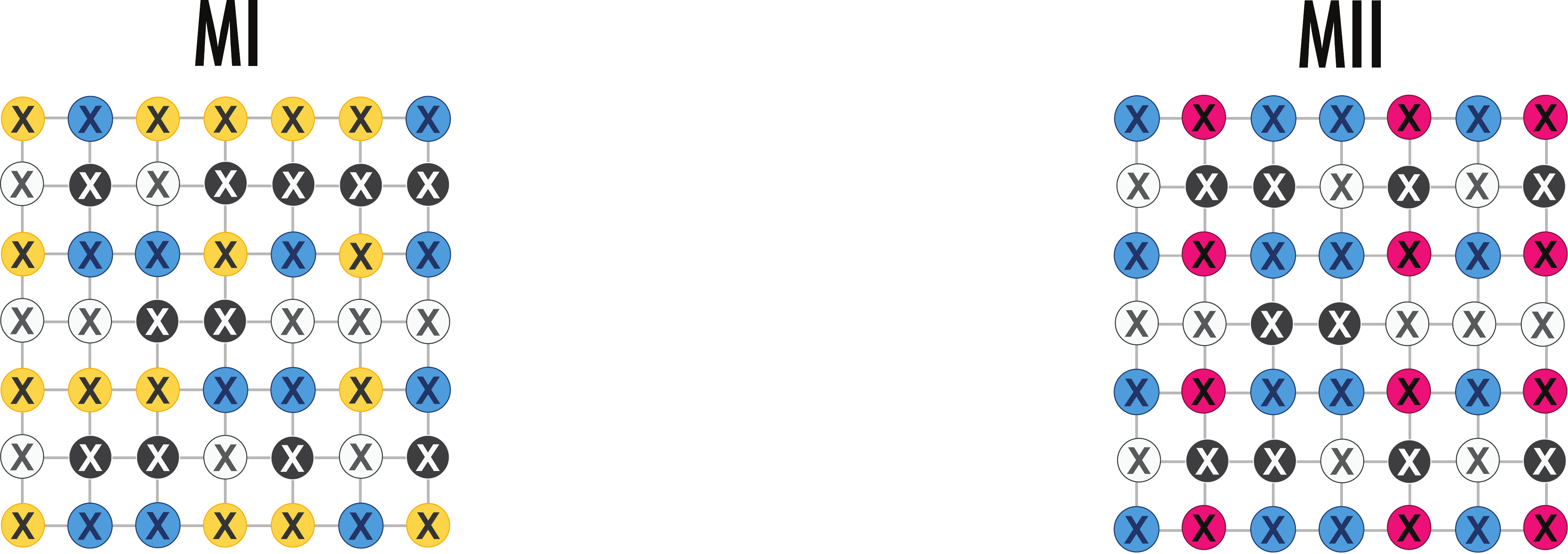}
\caption{Modified architectures \ref{enum:Arc1}-\ref{enum:Arc2}, named ``MI'' and ``MII'' in the figure. Changes are introduced on the even rows with respect to Fig.\mot{}\ref{fig:Experiment}. On even rows we pick $\beta_i\in\{0,\uppi/2\}$ (black stands for $0$, white stands for $\uppi/2$) and perform a local Hadamard rotation. In architecture MI, $\beta_i$ is uniformly-random. In architecture MII, $\beta_i$ is translation invariant on columns with period less than or equal to $4$ (i.e., TI$_{(4,\infty)}$ in the notation of the main text). }\label{fig:VariationArc1}
\end{figure}

In the modified architecture~\ref{enum:Arc2}, observation (i) and postselection of byproduct operators (Fig.\ \ref{fig:PostSelectingAngles}) yields TI$_{(2,\infty)}$-symmetric  circuits of  $\{H_i,\euler^{\pm\imun \frac{\uppi}{16} Z_i}, \euler^{\pm\imun \frac{\uppi}{16} X_i}\}$  gates. In combination with  the gadgets in the proof of Lemma~\ref{lemma:lemma:MBQCX-YOnColumns}, this lets us implement $k$-qubit TI$_{(2,\infty)}$-symmetric nearest-neighbor circuits of $\euler^{\mp \imun \frac{\uppi}{8} X_i X_{i+1}}$ gates.  Furthermore, byproduct operators also let us break the TI$_{(2,\infty)}$ symmetry via the identities below, and allow us to implement non-TI arbitrary nearest-neighbor Clifford+$T$ circuits as in the previous case with constant overhead:
\begin{align}
&\prod_{i=1}^{\lfloor k/2\rfloor} \euler^{\imun a_i \frac{\uppi}{4} X_{[i]} X_{[i+1]}}
=
\left(\prod_{i=1}^{\lfloor k/2\rfloor} \euler^{\imun  \frac{\uppi}{8} X_{[i]} X_{[i+1]}}\right) \left(\prod_{i=1}^{\lfloor k/2\rfloor} Z_{[i]}^{a_i}\right)
\left(\prod_{i=1}^{\lfloor k/2\rfloor} \euler^{-\imun  \frac{\uppi}{8} X_{[i]} X_{[i+1]}}\right),
& [i]:=2i-x_0,x_0\in\{0,1\}, a_i\in\{0,1\},\notag\\
&\prod_{i=1}^{\lfloor k/2\rfloor} \euler^{\imun b_i \frac{\uppi}{8} X_{[i]} }
=
\left(\prod_{i=1}^{\lfloor k/2\rfloor} \euler^{\imun  \frac{\uppi}{16} X_{[i]} }\right) \left(\prod_{i=1}^{\lfloor k/2\rfloor} Z_{[i]}^{b_i}\right)
\left(\prod_{i=1}^{\lfloor k/2\rfloor} \euler^{-\imun  \frac{\uppi}{16} X_{[i]} }\right),
& [i]:=2i-x_0,x_0\in\{0,1\}, b_i\in\{0,1\}.\notag
\end{align}
Again, this yields an  efficient and exact postselected implementation of the desired IQP circuits on $k$ logical qubits.

\section{Certification protocol}
\label{app:certification} 

\newcommand{\supp}{\mathrm{supp}}

\newcommand{\av}[1]{\langle #1\rangle}
\newcommand{\Av}[1]{\left\langle #1\right\rangle}
\newcommand{\abs}[1]{\lvert #1\rvert}
\newcommand{\Abs}[1]{\left\lvert #1\right\rvert}
\newcommand{\Norm}[1]{\left\lVert#1\right\rVert}
\newcommand{\iprod}[2]{( #1, #2)}
\newcommand{\Iprod}[2]{\left( #1, #2 \right)}
\newcommand{\set}[1]{\{ #1  \}}
\newcommand{\Set}[1]{\left \{ #1 \right \}}

\newcommand{\Pb}{\mathbb{P}}

\newcommand{\ee}{\mathrm{e}}
\newcommand{\ii}{\mathrm{i}}
\newcommand{\dd}{\mathrm{d}}

In this appendix, we describe an efficient  parallelizable certification protocol for ground states of gapped local Hamiltonians. The protocol is an optimization of the algorithm presented in Ref.\ \cite{Hangleiter} 
featuring a reduced sample complexity. A direct implementation of the algorithm of Ref.\mot\cite{Hangleiter} requires a super-cubic number $O(N^3\log N)$ of prepare-and-measure experiments, where $N$ is the system size. Our protocol brings this complexity down to $O(N^2)$ by combining parallel sequences of local measurements and efficient classical post-processing. Additionally, it relies on on-site single-qubit observables only, as opposed with few-body Hamiltonian terms, which are needed in Ref.\mot\cite{Hangleiter}.  These improvements render the  protocol of Ref.\mot\cite{Hangleiter} faster and more suitable for, e.g, near-term quantum speedup experiments. However, our optimized algorithm can be used for verifying ground states of arbitrary non-degenerate gapped local  Hamiltonians, and may also be of independent interest.

\subsection{Parent Hamiltonians}\label{app:parentHamiltonians}

As a preliminary, we derive non-degenerate gapped local parent Hamiltonians for the pre-measurement state vectors $\ket{\psi_\beta}$ of our architectures. To find the stabilizer operators corresponding to $\ket{\psi_\beta}$
in architectures\mot{}\ref{enum:Arc1}-\ref{enum:Arc3}, we note that $\ket{\psi_\beta}$ can be prepared 
by applying the tensor-product unitary $U_\beta = \prod_{i \in V} \euler^{-\beta_i Z_i/2}$
followed by the local quench unitary $U$ to the initial state $ \ket{+}_V = \prod_{i \in V} \ket{+}_i$ where $\ket{+} = (\ket{0} + \ket{1})/\sqrt{2}$ as
$\ket{\Psi_\beta} = U U_\beta \ket{+}_V$. 
Since $\ket{+}$ is a $+1$-eigenvector of $X$, the parent Hamiltonian of
$\ket{\Psi_\beta}$ is given by $- U U_\beta ( \sum_{i \in V} X_i)
U_\beta^\dagger U^\dagger$.
For architectures \ref{enum:Arc1}-\ref{enum:Arc2} this evaluates to 
\begin{align}
    H_{\textnormal{\ref{enum:Arc1},\ref{enum:Arc2}}} & = -\sum_{i \in V}  \left( X_{\beta_i,i} \prod_{j:
(i,j) \in E} Z_j \right) \, . 
\label{app:stabilizer-1and2}
\end{align}
For the specific case of architecture \ref{enum:Arc3}, 
let us partition the vertices of the lattice $\mathcal{L}$ into two sets $V_1$ and $V_2$, such that $V_1$
contains all qubits on the square primitive sublattice, and $V_2$ the remaining dangling-bond qubits.
Further, let $E_1$ contain those edges where $J_{i,j} = \pi/4$, and $E_2$ the
dangling bonds where $J_{i,j} = \pi/16$.
We find the corresponding parent Hamiltonian to be 
\begin{align}
    \label{arcIII}
    H_{\textnormal{\ref{enum:Arc3}}} = & - \sum_{\substack{i \in V_1\\k:(i,k)\in E_2}} \left( CT_{(i,k)} X_i CT_{(i,k)}^\dagger   \prod_{j: (i,j) \in
    E_1} Z_j \right)
     - \sum_{\substack{k \in V_2\\i:(i,k)\in E_2}} \left(   CT_{(i,k)} X_k CT_{(i,k)}^\dagger\right)  \, , 
\end{align}
where the two-body terms evaluate to 
\begin{equation}
\label{eq:FactorsofTwoBodyTerm}
CT_{(i,k)} X_i CT_{(i,k)}^\dagger =  (X_i - TX_iT^\dagger)
 Z_k + (X_i + T X _i T ^\dagger)
 \id_k.
\end{equation}
The Hamiltonian terms of $H_\mathnormal{\ref{enum:Arc3}}$ are 6-local except at the boundary where its locality is reduced. 

\subsection{Energy estimation of $G$-local Hamiltonians}

An $N$-qubit  non-degenerate  gapped local Hamiltonian $H=\sum_{i\in V} h_i$  with $\tau$-body interactions on  an (simple connected) interaction graph $G=(V,E)$ is called ``\emph{$G$-local}'' if each qubit is located at a vertex $i\in V$ and each term $h_i$ is supported on the neighborhood $\partial(i)$ of $i$,
\begin{equation}
\supp (h_i) \subseteq\partial(i):=\{j:(i,j)\in E \},\quad   \forall i\in V.
\end{equation}
Though we will consider arbitrary interaction graphs $G$ in our analysis, we will be particularly interested on \emph{constant-degree} ones in our applications: i.e., those with maximum vertex degree $\deg(G)$  upper bounded by a constant, independently of the number of qubits. Because, w.l.o.g, we can pick $\tau=\mathrm{deg}(G)$, the latter graphs model physical systems with geometrically constrained connectivity, which are ubiquitous both in condensed matter physics and quantum information processing. Examples of such graphs are lattices of constant geometric dimension $D$ and fixed-size primitive cells. In particular, the Hamiltonians of the main text are $\mathcal{L}$-local and have $\deg(\mathcal{L})=4$ ($\deg(\mathcal{L})=5$) for architectures \ref{enum:Arc1}-\ref{enum:Arc2} (\ref{enum:Arc3}). 

The key ingredient of our result below is a subroutine for estimating  average energies of $G$-local Hamiltonians using parallelized measurement circuits with time complexity dominated by chromatic number $G^{(2)}=(V^{(2)},E^{(2)})$ of the next-neighbor interaction graph of $G^{(2)}$. The latter has  with same vertices $V^{(2)}=V$ as $G$ and edges between all pairs of neighbors and next-neighbors of $G$ ($(v_1,v_2)\in E^{(2)}$ iff  $v_1,v_2\in V$ and have graph distance $d(v_1,v_2)\leq 2$). The existence of such parallel circuits relies on the existence of certain decompositions of local Hamiltonians into commuting terms, named ($\kappa, \alpha, \tau$) (``cat'') decompositions below. In short, a Hamiltonian is $(\kappa,\alpha,\tau)$ if it can be decomposed as a small sum Hermitian operators that admit parallel measurement circuits that use single-shot $\tau$-body operators. 
\begin{definition}[($\kappa, \alpha, \tau$)-decomposition]\label{def:CABDecomp}
A $G$-local $N$-qubit Hamiltonian $H$ on a graph $G=(V,E)$ is $(\kappa,\alpha,\tau)$-decomposable if there exist  constants $\alpha\in(0,1]$, $\kappa,\tau,\chi\in\mathbb{N}$, a partition $V=\bigcup_{i=1}^{\chi} V_i$, $\chi \leq \kappa$, and a map $f:[\kappa]\rightarrow [\chi]$,  such that
\begin{equation}\label{eq:RenormalizationJointMeasurability}
H=\sum_{i=1}^{\kappa} H_i, \:\: H_i := \sum_{j\in V_{f(i)}} h_{j}^{(i)}, \:\: \supp\left(h_j^{(i)}\right) \subset \partial(j),
\end{equation}
where (a) $\max_i{|V_i|}\leq \alpha N$, (b) the terms $h_j^{(i)}$ are hermitian, and (c) the energy distribution of $H_i$ can be sampled from via parallel measurements of $\tau$-body observables and efficient classical post-processing.
\end{definition}
A sufficient condition for Def\mot\ref{def:CABDecomp}.(c) to  hold is that $i\in[\kappa]$, all terms in $\{ h_j^{(i)}\}_j$ are $\tau$-local have non-overlapping support. Then, one  can simply measure all  $h_j^i$ in parallel, obtain  the outcomes $\{e_{j}^{(i)}\}_j$, and sample the output distribution of $H_i$ by classically computing the sum $e_{i}:=\sum_{j\in V_{f(i)}} e_{j}^i$. The fact that Hamiltonians of these form admit parallel measurement circuits is formalized by the following lemma, which generalizes Lemma 1 in Ref.\ \cite{Hangleiter}.

\begin{lemma}[Estimation of the energy] \label{lemma:energy}Let $H$ be an $N$-qubit Hamiltonian with a given $(\kappa,\alpha,\tau)$-decomposition (\ref{eq:RenormalizationJointMeasurability}). Let 
$P_{i,\mu}$  be the $\mu$th  $e_{i,\mu}$-eigenprojector of $H_i$, for $i\in[\kappa]$, and $X_i^{(j)}$ be the random variable that takes the value $e_{i,\mu}$ with probability $\tr (\rho_p^{(j)} P_{i,\mu})$, modeling a measurement of $H_i$ on the $j^{\text{th}}$ copy of $\rho_p$. Moreover, let 
$\av{H_i}_{\rho_p}^* = \frac{1}{m} \sum_{i=1}^m X_i^{(j)}$
be the estimate
of $\av{H_i}$ on $\rho_p$ by a finite-sample average of $m$ measurement
outcomes, and $\av{H}_{\rho_p}^* = \sum_{i=1}^\kappa \av{H_i}_{\rho_p}^*$ the resulting estimate of $\av{H}_{\rho_p}$.
Last, let $J = \max_\lambda \norm{ h_\lambda}$. Then, for any $\overline{p_{\mathrm{err}}}\in[1/2,1)$ and $\epsilon >0$ it holds that 
\begin{equation}
\Pb \left[\abs{\av{H}_{\rho_p}^* - \av{H}_{\rho_p} } \leq \epsilon \right] \geq \overline{p_{\mathrm{err}}} \, ,
\qquad 
\textnormal{whenever}
\qquad
 m \geq  \left[\frac{\alpha^2\kappa^2  J^2}{2 \epsilon^2} \ln\left[\frac{\kappa+1}{\ln(1/\overline{p_{\mathrm{err}}})}\right] \right]N^2.\label{measurements}
\end{equation}
\end{lemma}	
\begin{proof}
Since the random variables $\set{X_i^{(j)}}_j$ are independent and $0 \leq X_i^{(j)} \leq \norm{H_i}$,  Hoeffding's inequality implies,
\begin{align} \notag
\forall i \in [\kappa]: \, \Pb \left[ \abs{\av{H_i}_{\rho_p}^* -  
\av{H_i}_{\rho_p} } \leq \epsilon \right] \geq 1 - 2 \ee^{-\frac{2m \epsilon^2}{\norm{H_i}^2}}, 
\end{align}
and since all measurements are independent
\begin{align}
\Pb\left[\abs{\av{H}_{\rho_p}^* - \av{H}_{\rho_p} } \leq
\epsilon \right] &\geq \Pb  \left[  \forall i \in [\kappa]:  \abs{\av{H_i}_{\rho_p}^* -
 \av{H_i}_{\rho_p} } \leq \frac{\epsilon}{\kappa} \right] 
  \geq \prod_{i=1}^{\kappa} \left( 1 - 2 \ee^{-\frac{ 2 m
 \epsilon^2}{\kappa^2\norm{H_i}^2 }} \right)  \geq \left( 1 - 2 \ee^{-\frac{ 2 m
  \epsilon^2}{\kappa^2(\alpha NJ)^2} }  \right)^{\kappa} \geq \overline{p_{\mathrm{err}}} \, . \notag
\end{align}
The latter identity holds whenever $ m \geq m_{\mathrm{opt}}$ with
\begin{align}
 m_{\mathrm{opt}}:=\frac{\alpha^2\kappa^2  J^2 }{2 \epsilon^2} \ln\left[ \frac{2}{1 -
   \overline{p_{\mathrm{err}}}^{1/\kappa} }\right] N^2\quad \stackrel{\textnormal{\cite{Aolita-NatComm-2015}}}{\leq}\quad\left[\frac{\alpha^2\kappa^2  J^2}{2 \epsilon^2} \ln\left[\frac{\kappa+1}{\ln(1/\overline{p_{\mathrm{err}}})}\right] \right]N^2.
\end{align} 
\end{proof}
The next result states that any $G$-local Hamiltonian admits a $(\kappa,\alpha,\tau)$ decomposition.
\begin{lemma}[Local decompositions]\label{lemma:KMeasurability} Let $H$ be an $N$-qubit   $\tau$-body local Hamiltonian on a graph $G=(V,E)$, and let  $\tau\in O(1)$. 
Let  $\chi\left(G^{(2)}\right)$ and $\iota\left(G^{(2)})\right)$, denote the chromatic and independence numbers of $G^{(2)}$ \cite{bondy}. Then, $H$ admits a $(\kappa,\alpha,\tau)$    decomposition with $\kappa\leq \chi(G^{(2)})\leq \tau^2+1$, $\alpha\leq \iota\left(G^{(2)}\right)/N<1$. Furthermore, any $(\kappa,\alpha,\tau)$ decomposable Hamiltonian admits an on-site $(\kappa 4^{\tau}, \alpha ,1)$ decomposition.
\end{lemma} 
The bounds in Lemma\mot\ref{lemma:KMeasurability} are not necessarily tight for commuting Hamiltonians. For instance, the fully-connected Ising Hamiltonian $H=-\sum_{i,j}J_{i,j}Z_iZ_J-\mu\sum_k h_kZ_k$  admits a $(1,1,1)$ decomposition, though $\chi(G^{(2)})=N$, because it can be measured directly via single-shot on-site $Z_i$ measurements and classical post-processing (as it is a polynomial of the latter commuting observables).  For the Hamiltonians in architectures \ref{enum:Arc1}-\ref{enum:Arc3}, we also find much tighter bounds.

\begin{lemma}[Local decompositions in the architectures]\label{lemma:ArcHamiltonianDecomposition}
Let $H_{\text{a}}$ be the $N$-qubit Hamiltonian in architecture $a\in\{\textnormal{\ref{enum:Arc1},\ref{enum:Arc2},\ref{enum:Arc3}}\}$. Then,
\begin{itemize*}
\item $H_{\textnormal{\ref{enum:Arc1}}}$ admits a $\left(2,\tfrac{5}{9},1\right)$ decomposition of form (\ref{eq:RenormalizationJointMeasurability}) where every $H_i$ has on-site terms and $\mathrm{DO}$ symmetry.
\item $H_{\textnormal{\ref{enum:Arc2}}}$ admits a $\left(2,\tfrac{5}{9},1\right)$ decomposition of form (\ref{eq:RenormalizationJointMeasurability}) where every $H_i$ has on-site terms and $\mathrm{TI}_{(1,\infty)}$ symmetry.
\item $H_{\textnormal{\ref{enum:Arc3}}}$ admits  a two-body $\left(2,\tfrac{5}{9},2\right)$ and  an on-site $\left(32,\tfrac{5}{9},1\right)$ decompositions  where  every $H_i$ has $\mathrm{TI}_{(\sqrt{2},\sqrt{2})}$ symmetry.
\end{itemize*}
\end{lemma}
\begin{proof}[Proof of Lemma \ref{lemma:KMeasurability}]
We  prove the existence of the   $(\kappa,\alpha,\tau)$ decomposition by  constructing a partition $V=\bigcup_{i=1}^\kappa V_i$ such that
\begin{equation}\label{eq:Renormalization}
H=\sum_{i=1}^{\kappa} H_i, \quad H_i := \sum_{i\in V_i} h_{i},
\end{equation}
where the terms in $\{h_i\in H_i\}$  are $\tau$-body by construction and  have non-overlapping support  for all $i\in[1,\kappa]$, (hence, can be simultaneously measured).  First, $\chi\left(G^{(2)}\right)$ is the minimal number of classes in any vertex coloring of $G^{(2)}$ (i.e., a vertex partition where no pair of adjacent vertices falls in the same class).  
Further, two vertices $v_1,v_2\in V^{(2)}$ are adjacent in $G^{(2)}$ iff they are neighbors or next-neighbors in $G$. Letting $V=\bigcup_{i=1}^{\chi(G^{(2)})} V_i,$ be a minimal vertex coloring of $G^{(2)}$, we obtain a decomposition of form (\ref{eq:Renormalization}) with $\kappa\leq \chi\left(G^{(2)}\right)$. Last, picking $\alpha :=\max_i |V_i|/N$ we get $\alpha\leq\iota(G^{(2))}/ N$ since $G^{(2)}$ is connected.

The existence of the  $(\kappa4^{\deg(G)},\alpha,1)$ decomposition now follows by expanding  each term $h_i$ in every $H_i$ in (\ref{eq:Renormalization}) in a product basis 
\begin{equation}
 \mathcal{A}_i:=\left\{\bigotimes_{j\in\partial(i)} A(x_{j})_{j}\,\,A(x_{j})\in\mathcal{A}\right\}, 
 \end{equation}
 where $\mathcal{A}=\{A(\mu)\}_{\mu=1}^4$  
 is some Hermitian single-qubit operator basis (e.g., the standard Pauli matrices). By picking a fixed ordering of every set $\partial(j)$,
 this lets us write each $H_i$ as a  sum  of $4^{\tau}$ Hermitian operators  $\{ H_{i,x}, x=(x_1,\ldots,x_{\tau}),x_i=[1,4] \}_x$,  $H_{i,x} := \sum_{j\in V_i} \alpha_{j,x} \bigotimes_{k\in\partial(j) }A(x_{k})_{k}$, the energy distribution of which can be sampled from via on-site $A(x_i)$ measurements.


Finally, we upper bound  $\chi(G^{(2)})$. Let  $\deg(G)$ denote the maximum vertex degree of  $G$. Since every vertex  $v\in V^{(2)}$ has at most  has $\deg(G)$ neighbors and at most $\deg(G)(\deg(G)-1)$ next-neighbors, it follows that $\deg\left(G^{(2)}\right)\leq \deg(G)^2$. By Brook's theorem\mot\cite{bondy},  $\chi\left(G^{(2)}\right)\leq \deg(G)^2+1$. Last, we use that  $\tau=\deg(G)$, since the $\tau$ is by definition the maximum vertex degree of the interaction graph.
\end{proof}

\begin{proof}[Proof of Lemma \ref{lemma:ArcHamiltonianDecomposition}]

For $H_\textnormal{\ref{enum:Arc1},\ref{enum:Arc2}}$, we first split the Hamiltonian terms in (\ref{app:stabilizer-1and2}) into two groups (``even" and ``odd'') using a bi-coloring  of the $N$-sites square lattice, and set $H=H_{\textnormal{even}}+H_{\textnormal{odd}}$. The terms $\{h_i\}_i$ of  $H_{\textnormal{even}}$ ($H_{\textnormal{odd}}$) are  products of $X_{\tau_i,i}$, $Z_j$ on-site factors. Because we use a 2-coloring, the on-site factor list associated to two distinct $h_i,h_j$ terms contains at most two overlapping on-site pairs of form $(Z_k,Z_k),(Z_k',Z_k')$. Hence, overlapping terms are  identical and can be measured jointly. This allows us to measure $H_{\textnormal{even}}$ ($H_{\textnormal{even}}$) via a parallel measurement of all on-site factors and classical post-processing. This yields a $(\kappa,\alpha,\tau)$ decomposition with $\kappa=2$ and $\tau=1$. Further, the largest component $V_{\max}$ of a square lattice 2-coloring  has $N/2$ vertices for even $N$, and $(N+1)/2$ otherwise. Hence, we can pick  $\alpha=|V_{\max}|/N\leq 1/2(1+1/N)\leq 5/9$, where we use  that the smallest odd value of $N$ is 9. The same approach works leads to a $(2,5/9,2)$ decomposition for  $H_\textnormal{\ref{enum:Arc3}}$,  Eq.\ (\ref{arcIII}) where  $H_{\textnormal{even}}$, $H_{\textnormal{odd}}$ are  sums of products of on-site terms supported on the primitive lattice, and two-body terms acting on dangling-bonds. This leads to a $(32,5/9,1)$ one by expanding the 2-body terms on a local basis, as in the proof of Lemma\mot\ref{lemma:KMeasurability}. Finally, all new Hamiltonians inherit the symmetry of their corresponding parent Hamiltonian by construction.\qedhere
\end{proof}

\subsection{Certification protocol}

Finally, we  describe a quadratic-time weak-membership  certification protocol for ground states of  non-degenerate gapped $G$-local $\tau$-body Hamiltonians with constant $\tau$, and, in general, any $(\kappa,\alpha,\tau)$-measurable Hamiltonians (Def.\mot\ref{def:CABDecomp}). By virtue of Lemma\mot\ref{lemma:ArcHamiltonianDecomposition}, the protocol can be applied to efficiently certify the final state preparation of our proposed quantum architectures\mot\ref{enum:Arc1}-\ref{enum:Arc3}. We describe the protocol for the latter class since we know any $G$-local Hamiltonian is of that form (Lemmas\mot\ref{lemma:KMeasurability},\ref{lemma:ArcHamiltonianDecomposition}). The protocol is simply a parallelized version of the one in Ref.\mot\cite{Hangleiter}.
\begin{definition}[Weak-membership quantum state certification \cite{Aolita-NatComm-2015}]. Let $F_T>0$ be a threshold fidelity and $0<p_\mathrm{err}<1$ be a maximal failure probability. A test which takes as an input a classical description of $\rho_0$ and copies of a preparation of $\rho_p$, and outputs ``reject” or ``accept” is a weak-membership certification test if with high probability $p_\mathrm{succ}\geq 1 -p_\mathrm{err}$ it rejects every $\rho_p$ for which $F(\rho_p,\rho_0) \leq F_T$, and accepts every $\rho_p$ for which $F(\rho_p,\rho_0)\geq F_T+\delta$ for some fidelity gap $\delta>0$.
\end{definition}
\begin{protocol}[Certification of $(\kappa,\alpha,\tau)$-decomposable Hamiltonians]\label{certprotocol} The protocol receives 
a description of a non-degenerate gapped Hamiltonian $H$ that admits a  $(\kappa,\alpha,\tau)$-decomposition of form (\ref{eq:RenormalizationJointMeasurability}), which is given to us, and performs the following steps:
\begin{enumerate*}
\item Arthur chooses a threshold fidelity $F_T< 1$, maximal failure probability $1> p_{\mathrm{err}} >0$ and an  error $\epsilon \leq (1-F_T)/2$.
\item Arthur asks Merlin to prepare a sufficient number of copies of the
  ground state $\rho_0$ of $H$. 
\item Arthur performs $m$ energy measurements for each Hamiltonian term $H_i$ on distinct copies of
the state $\rho_p$ 
to determine an estimate $E^*$ of the
expectation value $\sum_i \tr[\rho_p H_i]$, with  $m$ given by expression \eqref{measurements}. Each $H_i$ is measured by a single-shot circuit of $\tau$-local observables and classical postprocessing..

\item From the estimate $E^*$ he obtains an estimate $F_{\min}^*$ of 
lower bound $F_{\min}=1-\av{H}_{\rho_p}/\Delta$  \cite{Hangleiter} on the fidelity $F =F(\rho_p, \rho_0)$
such that $F_{\min}^* \in [F_{\min} - \epsilon, F_{\min} + \epsilon]$ with probability at least $1-p_{\mathrm{err}}$. 

\item  If $F_{\min}^* < F_T + \epsilon$ he rejects, otherwise he
accepts. 
\end{enumerate*}
\end{protocol}
\begin{lemma}[Weak-membership certification] \label{certification}
Let $H$ be an $N$-qubit non-degenerate gapped $(\kappa,\alpha,\tau)$-decomposable Hamiltonian with known ground state energy $E_0$, gap $\Delta$,  interaction strength $J  = \max_\lambda \norm{h_\lambda}$, and let $E_0, \Delta^{-1}, J$ be upper bounded by a constant. 
Then, Protocol~\ref{certprotocol} is a weak-membership certification test, in the sense of \cite{Hangleiter},  with fidelity gap
\begin{align} 
\delta =  \left( 1- F_T \right) \left( 1- \frac{\Delta}{\norm{H}} \right) +
\frac{2 \epsilon \Delta}{\norm{H}} 
\, , \label{fidelity gap}
\quad 
\textnormal{and requires}
\quad
  m \geq  \left[\frac{\alpha^2\kappa^2  J^2}{2 \Delta^2\epsilon^2} \ln\left[-\frac{\kappa+1}{\ln(1-p_{\mathrm{err}} )}\right] \right]N^2
\end{align}
repetitions  to determine the expectation value $\av{H}_{\rho_p} =\tr[ H \rho_p]$. 
\end{lemma} 
In combination with Lemma\mot\ref{lemma:ArcHamiltonianDecomposition}, it follows that one can efficiently certify the final state preparation of our architectures \ref{enum:Arc1}-\ref{enum:Arc3}, if the latter are
at least $1/N$ close to the target state in fidelity since $\norm{H} \sim N $ and $\Delta $ is larger than a constant by construction of the parent Hamiltonians $H_{\textnormal{\ref{enum:Arc1}-\ref{enum:Arc3}}}$. 
 \begin{proof}[Proof of Lemma\mot\ref{certification}
\label{s:proof}]
The proof is identical to that of Proposition\mot{}1 in Ref.\ \cite{Hangleiter} if we substitute  Protocol 1 (Lemma 1) therein where our Protocol \ref{certprotocol} (our Lemma\mot\ref{lemma:energy}) in this appendix. We refer the reader to 
Ref.\ \cite{Hangleiter} for details.
\end{proof}

\section{Alternative weaker forms of Conjectures \ref{conj:ConjecturePolyHiear}-\ref{conj:IsingAverageComplexity}}\label{app:Conjectures}

In this appendix, we briefly discuss how our main result Theorem \ref{thm:Main} holds given even weaker versions of Conjectures\mot{}\ref{conj:ConjecturePolyHiear}-\ref{conj:IsingAverageComplexity}. First, in  Conjecture\mot\ref{conj:ConjecturePolyHiear}, it suffices for our purposes that the Polynomial Hierarchy does not collapse to its 3rd level, instead of being infinite. Second,  in Conjecture\mot\ref{conj:IsingAverageComplexity}, we do not need the problem of approximating Ising partition functions to be $\#\textsf{P}$-hard in average. Instead, it is enough that the this problem is not in the complexity class $\textsf{BPP}^\textsf{NP}$, which is contained in the 3rd level of the Polynomial Hierarchy (and would be in the 2nd, if the widely believed conjecture  $\textsf{P}=\textsf{BPP}$ \cite{aaronson2016} holds). Note how
this would be in striking contrast with our hardness result\mot\ref{lemma:SharpPHardness}, since the latter says that an oracle to solve the worst-case version of the same problem would allow us to solve \emph{all} problems in \emph{all} levels of the Polynomial Hierarchy. Third, if this weaker form of Conjecture\mot{}\ref{conj:IsingAverageComplexity} holds, then Conjecture\mot\ref{conj:ConjecturePolyHiear} is obviously not needed in the proof of Theorem\mot\ref{thm:Main}. 
We have also mentioned in the main text that stating Conjecture\mot{}\ref{conj:IsingAverageComplexity} in terms of relative errors---which is the approach followed here and in \cite{AaronsonArkhipov13LinearOptics4,BremnerJozsaShepherd08,Boixo161608.00263,SparseNoisySupremacy}---is somewhat more natural than stating it in terms of additive ones, as in Ref.\mot\cite{Gao17SupremacyIsing}. To illustrate the difference, note that there exist quantum algorithms for approximating (normalized) Ising partition functions up to polynomially small additive errors \cite{DelasCuevas11Q_AlgorithmsClassicalLatticeModels,MatsuoFujiiImoto14_Q_Alg_Partition_Functions}, while the latter are \textsf{\#P}-hard to approximate up to polynomially small and even constant relative ones.

\end{widetext}
\end{document}